\documentclass[11pt]{article}

\pdftrailerid{}
\newif\ifanon
\anonfalse

\usepackage[margin=1in]{geometry}
\usepackage{amsmath,amssymb,amsthm,mathtools}
\usepackage{graphicx}
\usepackage[table]{xcolor}
\usepackage{aliascnt}
\usepackage{enumitem}
\usepackage{booktabs}
\usepackage{array}
\usepackage[colorlinks=true,linkcolor=blue,citecolor=blue,urlcolor=blue]{hyperref}
\usepackage{cleveref}
\usepackage{tikz}
\usetikzlibrary{arrows.meta}
\usepackage{orcidlink}
\usepackage{placeins}
\usepackage[multiple]{footmisc}

\newif\ifshowcomments
\showcommentstrue
\ifshowcomments
\newcommand{\note}[1]{[{\footnotesize\sffamily\color{magenta}\textbf{Haoxing:} #1}]}
\newcommand{\cans}[1]{[{\footnotesize\sffamily\color{blue}\textbf{Answer:} #1}]}
\else
\newcommand{\note}[1]{}
\newcommand{\cans}[1]{}
\fi

\setlist[itemize]{leftmargin=2em}
\setlist[enumerate]{leftmargin=2em}

\theoremstyle{plain}

\newtheorem{theorem}{Theorem}[section]

\newaliascnt{lemma}{theorem}
\newtheorem{lemma}[lemma]{Lemma}
\aliascntresetthe{lemma}

\newaliascnt{proposition}{theorem}
\newtheorem{proposition}[proposition]{Proposition}
\aliascntresetthe{proposition}

\newaliascnt{corollary}{theorem}
\newtheorem{corollary}[corollary]{Corollary}
\aliascntresetthe{corollary}

\newaliascnt{conjecture}{theorem}
\newtheorem{conjecture}[conjecture]{Conjecture}
\aliascntresetthe{conjecture}

\newaliascnt{claim}{theorem}

\aliascntresetthe{claim}

\newaliascnt{fact}{theorem}

\aliascntresetthe{fact}

\theoremstyle{definition}

\newaliascnt{definition}{theorem}
\newtheorem{definition}[definition]{Definition}
\aliascntresetthe{definition}

\newaliascnt{example}{theorem}

\aliascntresetthe{example}

\newaliascnt{problem}{theorem}

\aliascntresetthe{problem}

\theoremstyle{remark}

\newaliascnt{remark}{theorem}

\aliascntresetthe{remark}

\crefname{theorem}{Theorem}{Theorems}
\Crefname{theorem}{Theorem}{Theorems}

\crefname{lemma}{Lemma}{Lemmas}
\Crefname{lemma}{Lemma}{Lemmas}

\crefname{proposition}{Proposition}{Propositions}
\Crefname{proposition}{Proposition}{Propositions}

\crefname{corollary}{Corollary}{Corollaries}
\Crefname{corollary}{Corollary}{Corollaries}

\crefname{conjecture}{Conjecture}{Conjectures}
\Crefname{conjecture}{Conjecture}{Conjectures}

\crefname{claim}{Claim}{Claims}
\Crefname{claim}{Claim}{Claims}

\crefname{fact}{Fact}{Facts}
\Crefname{fact}{Fact}{Facts}

\crefname{definition}{Definition}{Definitions}
\Crefname{definition}{Definition}{Definitions}

\crefname{example}{Example}{Examples}
\Crefname{example}{Example}{Examples}

\crefname{problem}{Problem}{Problems}
\Crefname{problem}{Problem}{Problems}

\crefname{remark}{Remark}{Remarks}
\Crefname{remark}{Remark}{Remarks}

\crefname{algorithm}{Algorithm}{Algorithms}
\Crefname{algorithm}{Algorithm}{Algorithms}

\crefname{appendix}{appendix}{appendices}
\Crefname{appendix}{Appendix}{Appendices}

\crefname{figure}{Figure}{Figures}
\Crefname{figure}{Figure}{Figures}

\newcommand{\AC}{\mathsf{AC}}
\newcommand{\kOV}{k\text{-}\mathsf{OV}}
\newcommand{\kXOR}{k\text{-}\mathsf{XOR}}
\newcommand{\kSUM}{k\text{-}\mathsf{SUM}}
\newcommand{\SUB}{\mathrm{SUB}}
\newcommand{\CLIQUE}{\mathrm{CLIQUE}}

\DeclareMathOperator{\size}{size}
\newcommand{\proj}{\mathrm{proj}}           

\newcommand{\F}{\mathbb{F}}
\newcommand{\Z}{\mathbb{Z}}
\newcommand{\N}{\mathbb{N}}

\newcommand{\bits}{\{0,1\}}
\newcommand{\1}{\mathbf{1}}
\newcommand{\xor}{\oplus}

\DeclareMathOperator{\vecop}{vec}

\DeclareMathOperator{\tw}{tw}

\newcommand{\COV}{C_{\mathrm{OV}}}
\newcommand{\Cxor}{C_{\oplus}}
\newcommand{\Csum}{C_{\Sigma}}
\newcommand{\Dzero}{D_0}
\newcommand{\mxor}{m_{\oplus}}
\newcommand{\msum}{m_{\Sigma}}

\newcommand{\idx}{\mathrm{idx}}
\newcommand{\Ical}{\mathcal{I}}

\title{Fine-Grained \texorpdfstring{$\AC^0$}{AC0} Lower Bounds for
\texorpdfstring{$\kOV$}{k-OV}, \texorpdfstring{$\kXOR$}{k-XOR}, and
\texorpdfstring{$\kSUM$}{k-SUM} \\ via Colored Subgraph Isomorphism}

\ifanon
  \author{}
\else
  \author{
    Haoxing Lin\,\orcidlink{0000-0001-9594-1871}\thanks{%
      This work was supported by the National Research Foundation, Singapore, under its NRF Fellowship programme, award no.~NRF-NRFF14-2022-0010.%
    }\\[0.3em]
    \normalsize National University of Singapore\\
    \normalsize \texttt{haoxingl@comp.nus.edu.sg}
  }
\fi
\date{}

\begin{document}
\maketitle

\begin{abstract}
We study the circuit cost of selecting $k$ objects satisfying orthogonality, zero XOR, or modular zero SUM. Our main lower bounds in nonuniform $\AC^0$ use a common framework of depth-zero projections from colored subgraph isomorphism. The projections use dimension, row count, or bit width $O(k\log n)$ and preserve circuit size, depth, and gate orientation. Stronger lower bounds for suitable source patterns transfer through the same framework to all three targets. Our XOR lower bounds require odd $k$.

For every fixed depth $d$ and sufficiently large fixed $k$, we prove unconditional lower bounds $n^{\Omega(k)}$ for OV and $(n/k)^{\Omega(k)}$ for XOR and SUM, with an absolute constant in the exponent independent of $d$ and $k$; only the threshold on $n$ may depend on $(d,k)$. For $k\to\infty$ with $k=n^{o(1)}$, combining our projections with Beame's small-clique lower bound yields $n^{\Omega_d(\min\{\sqrt{k},\log n\})}$ at every fixed depth $d$. For circuits of depth at most three with either type of output gate, we strengthen the exponent to $\Omega(k)$.

At dimension, row count, or bit width $\Theta(k\log n)$ with sufficiently large constant factors, all three problems admit depth-three circuits with a top OR gate. Their sizes match our lower bounds for unrestricted depth three up to constants in the exponent for sufficiently large $k$ and $n\ge k^C$ with a suitable constant $C>1$.

\end{abstract}

\clearpage
\setcounter{tocdepth}{2}
\tableofcontents
\setcounter{tocdepth}{2}

\clearpage
\section{Introduction}\label{sec:introduction}

How much circuit size is needed to select $k$ objects satisfying one joint relation? The three problems studied here express this witness-selection task in different algebras:

\begin{definition}[The three target problems]\label{def:target-problems}
\leavevmode
\begin{enumerate}[label=(\roman*),itemsep=2pt]
\item $\kOV$: For $n\le 2^D$, an input to $\kOV_{n,D,k}$ consists of $k$ ordered arrays $U_1,\ldots,U_k$, each containing $n$ vectors in $\bits^D$. Write $U_i=(u_{i,1},\ldots,u_{i,n})$, and the vectors within each array are pairwise distinct. The function accepts exactly when there are indices $j_1,\ldots,j_k\in[n]$ such that $\prod_{i=1}^k u_{i,j_i}(t)=0$ for every $t\in[D]$.
\item $\kXOR$: An input to $\kXOR_{n,m,k}$ is an indexed list $(a_1,\ldots,a_n)\in(\F_2^m)^n$. The function accepts exactly when there is a set $S\in\binom{[n]}{k}$ such that $\bigoplus_{i\in S}a_i=0^m$.
\item $\kSUM$: An input to $\kSUM_{n,m,k}$ is an indexed list $(z_1,\ldots,z_n)\in(\Z_{2^m})^n$. The function accepts exactly when there is a set $S\in\binom{[n]}{k}$ such that $\sum_{i\in S}z_i\equiv0\pmod{2^m}$.
\end{enumerate}
\end{definition}
\noindent We study how the exponent in the required circuit size grows with $k$. The computational model is nonuniform Boolean circuits with unbounded-fan-in AND/OR gates. Negations are at input leaves and circuit size counts non-input gates.

Our main contribution is a common projection framework that transfers lower bounds from colored subgraph isomorphism to all three problems. The projections need only $O(k\log n)$ bits per vector or number. They support both fixed $k$ and $k$ growing more slowly than every positive power of $n$. Any stronger source bound that applies in the required parameter range transfers to all three problems through the same projections.

The use of ``fine-grained'' concerns the dependence of the circuit-size exponent on $k$. Conditional running-time bounds motivate that question~\cite{ImpagliazzoPaturi01,ImpagliazzoPaturiZane01,VW18,ChenWilliams19} but do not by themselves constrain nonuniform circuits; they supply no hardness assumption here.

\subsection{Main results}\label{sec:intro-results}

There is a universal constant $C>0$ such that our lower bounds hold whenever each vector or number is represented using at least $Ck\log n$ bits. This representation length is the vector dimension $D$ for OV and the bit length $m$ for XOR and SUM.

\begin{table}[htbp]
\centering
\small
\renewcommand{\arraystretch}{1.08}
\setlength{\tabcolsep}{3pt}
\begin{tabular}{@{}>{\raggedright\arraybackslash}p{0.14\textwidth}>{\raggedright\arraybackslash}p{0.31\textwidth}>{\raggedright\arraybackslash}p{0.18\textwidth}>{\raggedright\arraybackslash}p{0.31\textwidth}@{}}
\toprule
Result & Regime / circuit class & Lower bound & Qualification \\
\midrule
\Cref{thm:B} & Sufficiently large fixed $k$, every fixed depth & $n^{\Omega(k)}$ & Universal exponent constant; threshold for $n$ depends on $(d,k)$ \\
\addlinespace[2pt]
\Cref{thm:C} & Growing $k$, every fixed depth $d$ & $n^{\Omega_d(\min\{\sqrt{k},\log n\})}$ & Unconditional \\
\addlinespace[2pt]
\Cref{thm:D} & Growing $k$, depth at most three & $n^{\Omega(k)}$ & Matches depth-three upper bounds up to constants in the exponent at $\Theta(k\log n)$ bits \\
\addlinespace[2pt]
\Cref{thm:E} & Growing $k$, fixed depth $d\ge4$ & $n^{\Omega_d(k)}$ & Conditional on \Cref{conj:dbd} \\
\bottomrule
\end{tabular}
\caption{Circuit-size lower bounds for all three targets under the common representation requirement. The growing-$k$ rows assume $k\to\infty$ and $k=n^{o(1)}$.}
\label{tab:intro-frontier}
\end{table}

\FloatBarrier

For growing $k$, we assume $k\to\infty$ and $k=n^{o(1)}$. Our XOR bounds require odd $k$, and OV retains the distinctness requirement in \Cref{def:target-problems}. \Cref{tab:intro-frontier} gives the lower bounds for sufficiently large $n$; the cited results state the full parameter ranges.

At representation length $\Theta(k\log n)$ with a sufficiently large constant factor, all three problems admit depth-three circuits with a top OR gate and size $n^{O(k)}$ (\Cref{tab:sigma3-upper-bounds}). For sufficiently large fixed $k$ and in the stated growing-$k$ range, these upper bounds match our unrestricted depth-three lower bounds up to constants in the exponent.

\FloatBarrier

\subsection{Technical Overview}\label{sec:intro-architecture}

Our main technical contribution is a projection framework that reduces colored subgraph isomorphism to each of $\kOV$, $\kXOR$, and $\kSUM$, so circuit lower bounds for the source transfer to all three targets.

\paragraph{The source problem.} The problem is specified by a pattern graph $P$ and an integer $N\ge1$. We consider each pattern vertex $v$ as a distinct \emph{color}. Its color class consists of the $N$ \emph{host vertices} $(v,1),\ldots,(v,N)$. Each pattern edge has an \emph{edge group} consisting of the $N^2$ possible \emph{host edges} between the corresponding color classes. The input has one bit for each possible host edge, indicating whether that edge is present in the \emph{host graph}. To specify a \emph{colored copy} of $P$, choose one host index $j(v)\in[N]$ for each color $v$ so that, for every pattern edge $(v, w)$, the host edge between $(v,j(v))$ and $(w,j(w))$ is present. Equivalently, a witness can choose one present host edge from each edge group, provided that the chosen edges use the same host index at every shared pattern vertex.

\paragraph{Encoding colored copies in the three problems.} We encode each possible host edge as a target object: a vector for OV or XOR, or an integer for SUM. An accepting target witness must choose one present edge per group and make the endpoint choices agree. OV puts the edge groups in separate arrays, so choosing one vector from each array already enforces one choice per group. XOR and SUM enforce this choice by equations within a single list. The XOR construction uses odd $k$; OV and SUM work for either parity (\Cref{sec:inst-ov,sec:inst-xor,sec:inst-sum}).

With $e(P)$ pattern edges, a witness selects $e(P)$ edge objects and one fixed auxiliary object, so we choose $e(P)=k-1$ to obtain $k$ selections. The separate OV arrays allow $N=\Theta(\sqrt n)$, whereas the single XOR or SUM list allows $N=\Theta(\sqrt{n/k})$. A lower bound that is a power of $N$ therefore becomes one that is a power of $n$ (\Cref{lem:host-substitutions}).

Every bit of the target instance is a constant, a source bit or its negation. Replacing the input leaves of a target circuit with these bits gives a source circuit with no increase in size or depth and the same top-gate type. This transfers any applicable source lower bound through the common projections (\Cref{thm:A}).

\paragraph{Source lower bounds.} For fixed $k$, we apply the Li--Razborov--Rossman (LRR) theorem to patterns built from bounded-degree expanders (\Cref{thm:B}). For growing $k$, the fixed-pattern theorem gives no uniform host threshold, so we use Beame's small-clique bound. A clique with at most $k-1$ edges has $O(\sqrt{k})$ vertices, while Beame's bound also limits the clique to $\log N$ vertices. These limits explain the exponent in \Cref{thm:C}.

At depth three, we reduce \emph{iterated intersection}, which asks whether the two sets in every input pair share an element, to colored subgraph isomorphism. The reduction transfers the top-OR lower bound of H{\aa}stad, Jukna, and Pudl{\'a}k~\cite[Eq.~(4.1)]{HJP95}. A second reduction and an extension of the same argument give the lower bound for circuits with a top AND gate (\Cref{thm:D}).

Finally, the Pattern-Uniform LRR conjecture asks for a source exponent linear in $k$ at polynomial host sizes. It would yield the remaining growing-$k$ bounds at every fixed depth $d\ge4$ through the same projections (\Cref{conj:dbd,thm:E}). \Cref{fig:intro-architecture} summarizes these transfers.

\begin{figure}[htbp]
\centering
\begin{tikzpicture}[
  x=1cm,y=1cm,>=Latex,
  every node/.style={align=center,font=\small},
  source/.style={draw,rounded corners,text width=2.85cm,minimum height=1.0cm,inner sep=2pt},
  family/.style={draw,rounded corners,minimum height=0.65cm,inner sep=3pt},
  interface/.style={draw,rounded corners,text width=8.7cm,minimum height=0.65cm,inner sep=3pt},
  target/.style={draw,rounded corners,text width=3.0cm,minimum height=0.6cm,inner sep=2pt}
]
\node[source] (fixed) at (-5.1,2.65) {Fixed-pattern LRR\\$\to$ Thm.~\ref{thm:B}};
\node[source] (intersection) at (-1.7,2.65) {Depth-three counting\\$\to$ Thm.~\ref{thm:D}};
\node[source,dashed] (uniform) at (1.7,2.65) {Pattern-Uniform LRR\\$\to$ Cor.~\ref{thm:E}};
\node[source] (beame) at (5.1,2.65) {Beame's clique bound\\$\to$ Thm.~\ref{thm:C}};
\node[family,text width=8.7cm] (expander) at (-1.7,1.35) {Patterns from bipartite expanders};
\node[family,text width=2.85cm] (clique) at (5.1,1.35) {Patterns from cliques};
\node[interface] (interface) at (0,0.15) {\Cref{thm:A}: common depth-zero projections};
\node[target] (ov) at (-4,-0.95) {$\kOV$};
\node[target] (xor) at (0,-0.95) {Odd-$k$ $\kXOR$};
\node[target] (sum) at (4,-0.95) {$\kSUM$};
\draw[->] (fixed.south) -- ([xshift=-34mm]expander.north);
\draw[->] (intersection.south) -- (expander.north);
\draw[->,dashed] (uniform.south) -- ([xshift=34mm]expander.north);
\draw[->] (beame.south) -- (clique.north);
\draw[->] (expander.south) -- ([xshift=-15mm]interface.north);
\draw[->] (clique.south) -- ([xshift=20mm]interface.north);
\draw[->] (interface.south) -- (ov.north);
\draw[->] (interface.south) -- (xor.north);
\draw[->] (interface.south) -- (sum.north);
\end{tikzpicture}
\caption{Source lower bounds pass through the same projections to all three targets. The dashed arrow marks the conditional use of \Cref{conj:dbd}.}
\label{fig:intro-architecture}
\end{figure}

\FloatBarrier

\subsection{Related work}\label{sec:related-work}

\paragraph{Direct Orthogonal Vectors lower bounds.} At polynomial dimension, Choudhury, Limaye, Sreenivasaiah, and Srinivasan give a stronger exponent for unrestricted fixed-depth circuits: for fixed $k,d\ge2$ and $\varepsilon>0$, their bound is $n^{k-\varepsilon}$ at $D=n^\varepsilon$~\cite[Thm.~4 and Footnote~4]{ChoudhuryLimayeSreenivasaiahSrinivasan26}. They state this for the complement of OV, which accepts when no orthogonal tuple exists. Complementation preserves size and depth in our circuit model, so the lower bound applies to OV as well. Their reduction already produces pairwise-distinct vectors~\cite[Thm.~11]{ChoudhuryLimayeSreenivasaiahSrinivasan26}. Our fixed-$k$ bound has a smaller exponent but holds at dimension $O(k\log n)$.

The closest $\Sigma_3$ result is due to Choudhury and Sreenivasaiah, who study OV without the distinctness requirement. For $k\le D$, they show that computing $\kOV$ as a disjunction of CNFs with at most $t$ literals per clause requires $\Omega((n/t)^k)$ CNFs feeding the top OR gate~\cite[Thm.~1.1]{ChoudhurySreenivasaiah26}. For fixed $k$ and $k\le t=O(1)$, this gives exponent $k$, matching exhaustive search over witnesses. Their generalization allows bottom gates that are monotone after fixing each input's polarity, with at most $t$ negated inputs~\cite[Thms.~1.2 and~5.12]{ChoudhurySreenivasaiah26}. They also use the iterated intersection bound to prove an unrestricted $\Pi_3$ lower bound of $2^{\Omega(n)}$ for $2$-OV at dimension $D=\Omega(n^2)$~\cite[Sec.~4]{ChoudhurySreenivasaiah24}.

Their $\Sigma_3$ bound for $t$-CNFs also covers growing $k$ at dimension $O(k\log n)$. At this dimension scale, \Cref{thm:D} gives an $n^{\Omega(k)}$ lower bound for both depth-three orientations, with unrestricted bottom fan-in, throughout its stated range of $n$ and $k$. Our OV bounds are proved on arrays with pairwise-distinct vectors. They also hold when repetitions are allowed, since any circuit correct on all arrays must be correct on those with distinct vectors. Our framework also gives bounds for XOR with odd $k$ and SUM for both parities.

\paragraph{Other computational models.} Kane and R.~Ryan Williams prove OV lower bounds for branching programs and for formulas with bounded fan-in or symmetric gates of unbounded fan-in~\cite[Thms.~2, 3, and~5]{KaneWilliams19}. Choudhury et~al.\ strengthen the bounds for branching programs and bounded-fan-in formulas~\cite[Thms.~6 and~25]{ChoudhuryLimayeSreenivasaiahSrinivasan26}. A branching program follows a path of input tests; a formula cannot share an intermediate result among several later gates. Our circuits allow such sharing. Kane--Williams count \emph{wires}, or connections, in their symmetric-gate theorem and depth-three conjecture; we count non-input gates~\cite[Thm.~5 and Conj.~10]{KaneWilliams19}.

Kane--Williams also give an average-case formula upper bound for a distribution in which the input bits are sampled independently~\cite[Thm.~6]{KaneWilliams19}. These formulas may err on exceptional inputs; our circuits must be correct on every input meeting the distinctness requirement.

Choudhury et~al.\ also prove lower bounds at logarithmic dimension for monotone circuits computing the complement of $2$-OV with repetitions allowed~\cite[Thm.~5]{ChoudhuryLimayeSreenivasaiahSrinivasan26}. Monotone circuits have no input negations; our lower bounds allow them.

\paragraph{Source lower bounds and broader context.} Classical $\AC^0$ lower bounds establish the limitations of constant-depth computation~\cite{Ajtai83,FSS84,Yao85,Hastad86}. Our fixed-depth results draw on the subgraph-detection line of Beame, Rossman, Amano, and Li--Razborov--Rossman~\cite{Beame90,Rossman08,Amano10,LRR17}. LRR's colored fixed-pattern theory and Rosenthal's average-case refinement~\cite{Rosenthal19} relate the circuit-size exponent to the structure of the pattern. We use the fixed-pattern LRR bound for \Cref{thm:B} and Beame's uniform small-clique bound for \Cref{thm:C}. Our conditional extension based on~\cref{conj:dbd} asks for an exponent linear in $k$ once the host size is polynomial in $k$, with constants depending only on depth. Published fixed-pattern statements do not provide that guarantee.

Other OV--circuit connections study what faster OV algorithms would imply for circuit satisfiability and circuit lower bounds~\cite{AbboudBringmannDellNederlof18,RRW25}. Our results directly lower-bound the size of circuits computing the witness-selection problems.

For XOR and SUM, generalized-birthday algorithms exploit random lists with many solutions~\cite{Wagner02}. Other work studies the time needed to find a solution in random lists or in lists constructed with a hidden solution~\cite{DinurKellerKlein24,AgrawalEtAl24}. Parity-counting hardness concerns whether the number of solutions is odd~\cite{DLSVW25}. Our question is the circuit size needed to decide whether a solution exists on every input. Our proofs use no running-time assumption from this work, and worst-case circuit hardness alone does not establish hardness on the random distributions used in cryptography.


\section{A Common Projection Framework}\label{sec:sources}\label{sec:framework}

We give a common way to encode colored subgraph isomorphism in the three targets of \Cref{def:target-problems}. We first fix the notation and circuit conventions, then explain the source problem and the checks shared by the projections. The next three subsections give the OV, XOR and SUM constructions. Finally, \Cref{sec:lower-bound-transfer} collects their guarantees and explains how source lower bounds become target lower bounds.

\subsection{Notations and preliminaries}\label{sec:preliminaries}\label{sec:target-interfaces}

\paragraph{Basic notation and parameters.} All logarithms are base two. For a nonnegative integer $r$, write $[r]:=\{1,\ldots,r\}$, with $[0]=\varnothing$, and let $\binom{[r]}{k}$ be the sets of $k$ distinct indices in $[r]$. Write $\F_2$ for the field $\bits$ with addition and multiplication modulo two, and $\Z_{2^m}$ for the integers modulo $2^m$. The symbol $\oplus$ denotes bitwise XOR, or addition in a vector space over $\F_2$. The vector $0^m$ has $m$ zero coordinates, and $\dot\cup$ denotes a union of disjoint sets.

The parameter $k$ is the number of objects selected by a target witness. Each OV array has $n$ vectors of dimension $D$; an XOR or SUM list has $n$ entries of $m$ bits each. The parameter $N$ counts host vertices in each source color class, as defined below.

\paragraph{Circuit conventions.} Circuits have unbounded-fan-in AND/OR gates and negations only at input leaves. Depth is the maximum number of non-input gates along an input-to-output path, and $\size_d(g)$ is the minimum number of non-input gates in such a circuit of depth at most $d$ computing $g$. We write $\Sigma_3=\mathsf{OR}\circ\mathsf{AND}\circ\mathsf{OR}$ and $\Pi_3=\mathsf{AND}\circ\mathsf{OR}\circ\mathsf{AND}$ for the two depth-three orientations, and $\size_{\Sigma_3}(g)$ and $\size_{\Pi_3}(g)$ for their minimum circuit sizes. Flattening adjacent gates of the same type gives
\begin{equation}\label{eq:depth-three-size}
\size_3(g)=\min\{\size_{\Sigma_3}(g),\size_{\Pi_3}(g)\}.
\end{equation}
A literal is an input bit or its negation. A clause is an OR of literals, and a CNF is an AND of clauses. For a CNF $T$, let $s(T)$ be its clause count. Write $M$ for the number of input variables. A $\Sigma_3$ circuit of size $S$ on $M$ variables can be written as $\bigvee_{i=1}^{\tau}T_i$ with
\begin{equation}\label{eq:sigma3-normal-form}
\tau\le S+2M,
\qquad
s(T_i)\le S+2M\quad(i\in[\tau]).
\end{equation}
Indeed, the top OR and each middle AND take inputs from at most $S$ gate outputs and $2M$ literals. The dual statement holds for $\Pi_3$.

\paragraph{Graph notation.} All graphs are finite and undirected, with no parallel edges or self-loops. For a graph $P$, write $V(P)$ and $E(P)$ for its vertex and edge sets, $v(P):=|V(P)|$, $e(P):=|E(P)|$, $\deg_P(a)$ for the number of edges incident to a vertex $a$, and $\Delta(P)$ for its maximum degree. A vertex is isolated if its degree is zero. We call the graph whose colored copy we seek the \emph{pattern graph}. We denote its edges by $f$ and fix an arbitrary order $f=(f_{\mathrm L},f_{\mathrm R})$ of each edge's endpoints.

\paragraph{Target conventions.} We use the targets defined in \Cref{def:target-problems}. The vectors within each OV array must be pairwise distinct. For $\kOV$, $\size_d$ denotes the minimum size of a circuit that is correct on every input satisfying the distinctness requirement. A circuit correct when repetitions are allowed must also be correct when the vectors are distinct, so our lower bounds apply to both variants. Inputs of XOR and SUM may carry equal values.

\paragraph{Depth-zero projections.}\label{par:projection} For Boolean functions $F$ and $G$ with specified input domains, write $F\le_{\proj}G$ if a map $R$ sends the domain of $F$ into that of $G$, satisfies $F(x)=G(R(x))$ for every $x$ in the domain of $F$, and has each output bit equal to a constant, an input bit, or its negation.

\paragraph{Projection monotonicity.}\label{par:projection-monotonicity} If $F\le_{\proj}G$, then $\size_d(G)\ge\size_d(F)$ for every depth $d$, also within a fixed orientation such as $\Sigma_3$ or $\Pi_3$. Indeed, substituting the projection's constants and literals at the input leaves of a circuit for $G$ gives a circuit for $F$: the map respects the input domains, and propagating constants cannot increase size or depth or change the required orientation.

\subsection{The source problem: colored subgraph isomorphism}\label{sec:structured-source}\label{sec:projection-skeleton}

As in the Introduction, the \emph{host graph} is the input graph in which we seek a colored copy of the pattern graph $P$. Its vertices, called \emph{host vertices}, are divided into one color class of size $N$ for each pattern vertex. For each pattern edge, the $N^2$ possible edges between its two color classes form an \emph{edge group}; we call them \emph{possible host edges}, and the input specifies which are present. The problem asks whether we can choose one host vertex from every color class so that every edge of $P$ is present between the chosen vertices. Such a choice is a \emph{colored copy} of $P$. Equivalently, we can choose one present host edge from each edge group and require the choices to agree wherever pattern edges share a vertex.

\begin{definition}[Colored $P$-subgraph isomorphism]\label{def:psub}
Fix a pattern graph $P$ and an integer $N\ge1$. We write $P\text{-}\SUB_N$ for the Boolean function deciding colored $P$-subgraph isomorphism. The host has vertex set $V(P)\times[N]$, with color class $\{a\}\times[N]$ for each $a\in V(P)$. The problem input consists of bits
\begin{equation}\label{eq:psub-coordinates}
X_{f,u,w}\in\{0,1\},
\qquad
\forall f=(f_{\mathrm L},f_{\mathrm R})\in E(P),
\qquad
\forall u,w\in[N].
\end{equation}
The variable $X_{f,u,w}$ records whether the colored host contains the edge between $(f_{\mathrm L},u)$ and $(f_{\mathrm R},w)$. The function accepts exactly when there is a map $\phi:V(P)\to[N]$ such that
\begin{equation}\label{eq:psub-acceptance}
X_{f,\phi(f_{\mathrm L}),\phi(f_{\mathrm R})}=1
\qquad\text{for every }f\in E(P).
\end{equation}
An accepting map $\phi$ specifies a \emph{colored copy} of $P$: its vertices are $(a,\phi(a))$ for $a\in V(P)$, and its edges are those required by \eqref{eq:psub-acceptance}.
\end{definition}

There are $M_{\mathrm{src}}:=e(P)N^2$ source input bits. To recover the map in \Cref{def:psub} from a choice of one present host edge per pattern edge with endpoint agreement, fix a non-isolated pattern vertex $a$. The host edges corresponding to the pattern edges incident to $a$ all meet at one host vertex $(a,u)$; set $\phi(a):=u$. Assign any index in $[N]$ to each isolated pattern vertex. Every edge required by \eqref{eq:psub-acceptance} is present, so this map is accepting. Multiple color classes may use the same index. Thus $\phi$ need not be injective into $[N]$, although the map $a\mapsto(a,\phi(a))$ into the host is always injective.

This is exactly the structured colored function of Li, Razborov, and Rossman~\cite[Sec.~5, p.~960]{LRR17}; \Cref{prop:psub-equivalence} gives the coordinate bijection. Their host parameter $n$ is our $N$. Our $n$ denotes the number of vectors in each OV array or the size of an XOR or SUM input.

\paragraph{The objects selected by a target witness.} For all three projections, fix an integer $k\ge2$, a pattern graph $P$ with $e(P)=k-1$, and a host size $N\ge1$.

Our target instance uses one \emph{edge object} to represent each possible host edge. The objects for one pattern edge form an \emph{edge group}. A target witness must select one edge object from each of the $k-1$ groups and one \emph{fixed auxiliary object}, making $k$ selections altogether.

For $f=(f_{\mathrm L},f_{\mathrm R})\in E(P)$ and $(u,w)\in[N]^2$, we use an edge object $c_{f,u,w}$ (corresponding to source bit $X_{f,u,w}$) to represent the possible host edge from $(f_{\mathrm L},u)$ to $(f_{\mathrm R},w)$. We will later show that each construction uses the absence bit $\overline X_{f,u,w}:=1-X_{f,u,w}$ to check edge presence, encoding this object as a vector for OV and XOR or an integer for SUM. Write $\alpha$ for the fixed auxiliary object. Together, the auxiliary and edge objects form the set
\begin{equation}\label{eq:genuine-object-skeleton}
\{\alpha\}\mathbin{\dot\cup}\bigcup_{f\in E(P)}\{c_{f,u,w}:(u,w)\in[N]^2\}.
\end{equation}
OV places the auxiliary object in its own array and each edge group in a separate array, giving $k$ arrays. XOR and SUM place the auxiliary object and all edge objects in a single indexed list.

\paragraph{The three source requirements.} The description of a colored copy by its host edges gives three checks that every construction must enforce.
\begin{itemize}[itemsep=2pt]
\item \textbf{Selection.} Choose exactly one edge object from every edge group, so every pattern edge is represented once.
\item \textbf{Edge presence.} Every selected $c_{f,u,w}$ must satisfy $X_{f,u,w}=1$ in the source input.
\item \textbf{Endpoint agreement.} Whenever pattern edges share a vertex $a$, their selected host edges must use the same host vertex index for $a$. Otherwise the chosen edges cannot form a colored copy.
\end{itemize}
Together, these checks say exactly that the selected edge objects describe a colored copy of $P$: endpoint agreement defines one host index for each non-isolated pattern vertex, and edge presence verifies every edge required by \Cref{def:psub}.

\paragraph{Two additional target requirements.} The source problem places no condition on the length of a target input or on whether its encodings are distinct. Our projections must also satisfy these requirements of the target problems.
\begin{itemize}[itemsep=2pt]
\item \textbf{Padding exclusion.} The target length $n$ need not exactly match the number of objects produced from a source of size $N$. In OV, for example, an edge array initially has $N^2$ vectors, whereas $n$ need not be a perfect square. Each edge array therefore needs $n-N^2$ padding vectors, and the auxiliary array needs $n-1$. XOR and SUM put all objects in a single list, so they need $n-1-e(P)N^2$ padding entries. We require $N^2\le n$ for OV and $1+e(P)N^2\le n$ for XOR and SUM, then fill the remaining positions with \emph{padding objects}. These extra objects must be excluded from accepting witnesses; naive zero padding should not be used, as $k$ zero padding entries would always be accepted in an XOR or SUM input list.
\item \textbf{Vector distinctness (OV).} Different host edges can receive identical vectors if we encode only the source checks above, and repeated padding vectors would also violate the OV input domain. We add identifiers that distinguish positions within each array while preserving which selections are orthogonal. XOR and SUM permit equal values at different list indices, so they need no such identifiers.
\end{itemize}

For odd-$k$ XOR and for SUM, the fixed auxiliary object is the key to force exactly one edge object per group (\Cref{sec:inst-xor,sec:inst-sum}). OV already enforces one vector per array by definition, so the auxiliary object has no role in enforcing one edge choice per group. Its array lets the same $(k-1)$-edge pattern give $k$ selections for all three targets.

\paragraph{Notation for endpoint comparisons.} To implement endpoint agreement, we compare the host indices at each shared pattern vertex. An edge object $c_{f,u,w}$ represents the possible host edge with endpoint indices $u$ and $w$; these indices are fixed by the object, independently of whether that edge is present in the source input.

For each pattern vertex $a$, define $\idx_a$ on triples $(f,u,w)$ with $f\in E(P)$ incident to $a$ and $(u,w)\in[N]^2$. It returns the index of the endpoint of color $a$ in the possible host edge represented by $c_{f,u,w}$:
\begin{equation}\label{eq:candidate-label}
\idx_a(f,u,w):=
\begin{cases}
u,&a=f_{\mathrm L},\\
w,&a=f_{\mathrm R}.
\end{cases}
\end{equation}
Put $\ell_N:=\lceil\log N\rceil$ and let $\operatorname{bit}_t(x)$ be the $t$-th bit of the $\ell_N$-bit encoding of $x-1$, for $x\in[N]$ and $t\in[\ell_N]$.

At each pattern vertex, the edge objects chosen for its incident pattern edges must encode the same host index to define a single source map $\phi$. It suffices to compare the object for each incident pattern edge with the object for one reference edge. Accordingly, fix one reference incident edge $\rho(a)$ for each non-isolated vertex $a$, and let
\begin{equation}\label{eq:comparison-incidences}
\Ical(P):=\{(a,g):a\in V(P),\ \deg_P(a)\ge1,\ g\ni a,\ g\ne\rho(a)\}.
\end{equation}
Thus $\Ical(P)$ indexes one comparison for every non-reference incidence, and
\begin{equation}\label{eq:comparison-count}
|\Ical(P)|=\sum_{a:\deg_P(a)\ge1}(\deg_P(a)-1)\le\sum_{a\in V(P)}\deg_P(a)=2e(P).
\end{equation}
Degree-one vertices contribute no non-reference incidence and zero to this sum.

The next three subsections implement these checks in the three target problems, using the same source instance as a running example. \Cref{thm:A} collects the conditions on $N$ and $n$ under which all objects fit in the target input, together with the required vector dimensions or bit lengths. \Cref{lem:host-substitutions} then transfers source lower bounds to the targets.

\paragraph{Running example.} We use a single source instance to illustrate the three constructions and their correctness checks. \Cref{fig:toy-source} shows the source, and \Cref{fig:toy-ov,fig:toy-xor,fig:toy-sum} show its OV, XOR and SUM encodings, respectively. Set $k=3$, $N=2$, and let $P$ be the path with edges $f_1=(p_1,p_2)$, $f_2=(p_2,p_3)$. For the reference comparisons in \eqref{eq:comparison-incidences}, choose $\rho(p_2)=f_1$. Thus $\Ical(P)=\{(p_2,f_2)\}$ and $\ell_N=1$: we compare the host index at $p_2$ encoded by an edge object for $f_1$ with that encoded by an edge object for $f_2$. Write $X^i_{uw}:=X_{f_i,u,w}$. The unique copy uses the present host edges recorded by $X^1_{21}=X^2_{12}=1$. OV uses array size $n=4$; XOR and SUM use list length $n=12$, leaving three padding indices. Repeated padding values are legal in XOR and SUM.

\begin{figure}[htbp]
\centering
\begin{minipage}[c]{0.48\textwidth}
\centering
\begin{tikzpicture}[every node/.style={font=\small},host/.style={circle,draw,inner sep=1.5pt,minimum size=5mm},pattern/.style={circle,draw,inner sep=1.5pt,minimum size=6mm},copy edge/.style={very thick},near miss/.style={thick,dash pattern=on 5pt off 1.5pt on 1pt off 1.5pt},absent edge/.style={densely dashed,black!50}]
\node[font=\small,anchor=east] at (-0.45,2.45) {$P$};
\node[pattern] (p1) at (0,2.45) {$p_1$};
\node[pattern] (p2) at (2,2.45) {$p_2$};
\node[pattern] (p3) at (4,2.45) {$p_3$};
\draw (p1) -- node[above] {$f_1$} (p2);
\draw (p2) -- node[above] {$f_2$} (p3);
\node[font=\scriptsize] at (0,1.55) {color class $p_1$};
\node[font=\scriptsize] at (2,1.55) {color class $p_2$};
\node[font=\scriptsize] at (4,1.55) {color class $p_3$};
\node[host] (h11) at (0,0.85) {$1$};
\node[host] (h12) at (0,0.05) {$2$};
\node[host] (h21) at (2,0.85) {$1$};
\node[host] (h22) at (2,0.05) {$2$};
\node[host] (h31) at (4,0.85) {$1$};
\node[host] (h32) at (4,0.05) {$2$};
\draw[absent edge] (h11) -- (h21);
\draw[absent edge] (h11) -- (h22);
\draw[copy edge] (h12) -- (h21);
\draw[absent edge] (h12) -- (h22);
\draw[absent edge] (h21) -- (h31);
\draw[copy edge] (h21) -- (h32);
\draw[near miss] (h22) -- (h31);
\draw[absent edge] (h22) -- (h32);
\end{tikzpicture}
\end{minipage}\hfill
\begin{minipage}[c]{0.48\textwidth}
\centering
\renewcommand{\arraystretch}{1.2}
\begin{tabular}{@{}c cccc@{}}
\toprule
$X^i_{uw}$ & $11$ & $12$ & $21$ & $22$ \\
\midrule
$i=1$ & $0$ & $0$ & $\mathbf{1}$ & $0$ \\
$i=2$ & $0$ & $\mathbf{1}$ & $1$ & $0$ \\
\bottomrule
\end{tabular}

\medskip
\small Example with a unique copy: $(2,1,2)$.
\end{minipage}
\caption{The example instance used in all three constructions, with $N=2$. Heavy edges form the unique copy $(2,1,2)$; the dash-dotted edge is present but disagrees at the shared endpoint, and gray dashed edges are absent. The table gives all eight source bits.}
\label{fig:toy-source}
\end{figure}

\FloatBarrier

\subsection{\texorpdfstring{$k$-Orthogonal Vectors}{k-Orthogonal Vectors}}\label{sec:inst-ov}\label{sec:instantiations}

Assume $N^2\le n$. We encode the objects in \eqref{eq:genuine-object-skeleton} as binary vectors: the auxiliary object $\alpha$ is represented by $v_\alpha$ in an array $U_\alpha$, and each edge object $c_{f,u,w}$ by $v_{f,u,w}$ in an array $U_f$, one array for each $f\in E(P)$. Fill every array to size $n$ with padding vectors. Once padding is excluded, choosing one vector per array gives the auxiliary vector and one edge vector per group.

A choice of one vector per array is orthogonal exactly when, in each coordinate, the product of the $k$ selected entries is zero (\Cref{def:target-problems}). We therefore design coordinates of these vectors so that an absent edge, endpoint disagreement or padding choice makes some coordinate equal to $1$ on all selected vectors. Additional identifier coordinates ensure distinctness without changing orthogonality. In \Cref{fig:toy-ov}, columns are vectors grouped by array and rows are coordinates; the starred columns form an orthogonal witness.

\paragraph{Edge presence and padding exclusion.}\label{par:ov-dummy-purge} We introduce one coordinate $q_\alpha$ to force the choice of $v_\alpha$ from $U_\alpha$. It is $0$ on $v_\alpha$ and $1$ on every other vector. Selecting $v_\alpha$ makes the product on $q_\alpha$ zero. Selecting a padding vector from $U_\alpha$ instead makes all $k$ selected entries on $q_\alpha$ equal to $1$, so those vectors cannot form an orthogonal witness.

For each pattern edge $f\in E(P)$, introduce one coordinate $q_f$ to check that the choice from $U_f$ represents a present edge, giving $e(P)=k-1$ such coordinates. Set $q_f$ to $\overline X_{f,u,w}$ in group $f$ and to $1$ everywhere else. If $v_{f,u,w}$ is selected, the other $k-1$ selected vectors contribute $1$, so the product on $q_f$ is $\overline X_{f,u,w}$. Choosing padding also gives product $1$, excluding padding from that array.

The second and fourth selections in the table of \Cref{fig:toy-ov} illustrate these rejections. In the second, $v_1^{11}$ represents an absent edge, so the product on $q_{f_1}$ is $1$; in the fourth, $\delta^1$ is padding in $U_\alpha$, so the product on $q_\alpha$ is $1$.

\begin{figure}[ht]
\centering
\begingroup
\small
\newcommand{\toyconst}[1]{\cellcolor{black!14}\ensuremath{#1}}
\newcommand{\toyone}{\toyconst{1}}
\newcommand{\toyzero}{\ensuremath{0}}
\newcommand{\toylit}[3]{\begingroup\setlength{\fboxsep}{.3pt}\setlength{\fboxrule}{.25pt}\fbox{\ensuremath{\overline{X}^{\,#1}_{#2}\!{}_{\scriptscriptstyle(#3)}}}\endgroup}
\newcommand{\toywit}[1]{\ensuremath{\boldsymbol{#1}\!{}^\star}}
\setlength{\tabcolsep}{1.6pt}
\renewcommand{\arraystretch}{1.10}
\begin{tabular}{@{}c|l|cccc|cccc|cccc@{}}
\toprule
 & & \multicolumn{4}{c|}{$U_\alpha$} & \multicolumn{4}{c|}{$U_{f_1}$} & \multicolumn{4}{c}{$U_{f_2}$} \\
Role & coordinate & \toywit{v_\alpha} & $\delta^1$ & $\delta^2$ & $\delta^3$ & $v_1^{11}$ & $v_1^{12}$ & \toywit{v_1^{21}} & $v_1^{22}$ & $v_2^{11}$ & \toywit{v_2^{12}} & $v_2^{21}$ & $v_2^{22}$ \\
\midrule
Padding & $q_\alpha$ & \toyzero & \toyone & \toyone & \toyone & \toyone & \toyone & \toyone & \toyone & \toyone & \toyone & \toyone & \toyone \\
Edges & $q_{f_1}$ & \toyone & \toyone & \toyone & \toyone & \toylit{1}{11}{1} & \toylit{1}{12}{1} & \toylit{1}{21}{0} & \toylit{1}{22}{1} & \toyone & \toyone & \toyone & \toyone \\
Edges & $q_{f_2}$ & \toyone & \toyone & \toyone & \toyone & \toyone & \toyone & \toyone & \toyone & \toylit{2}{11}{1} & \toylit{2}{12}{0} & \toylit{2}{21}{0} & \toylit{2}{22}{1} \\
Agreement & $C^{10}$ & \toyone & \toyone & \toyone & \toyone & \toyzero & \toyone & \toyzero & \toyone & \toyone & \toyone & \toyzero & \toyzero \\
Agreement & $C^{01}$ & \toyone & \toyone & \toyone & \toyone & \toyone & \toyzero & \toyone & \toyzero & \toyzero & \toyzero & \toyone & \toyone \\
Distinctness & $\mathrm{id}^{E}_1$ & \toyzero & \toyzero & \toyzero & \toyzero & \toyzero & \toyzero & \toyone & \toyone & \toyzero & \toyzero & \toyone & \toyone \\
Distinctness & $\mathrm{id}^{E}_2$ & \toyzero & \toyzero & \toyzero & \toyzero & \toyzero & \toyone & \toyzero & \toyone & \toyzero & \toyone & \toyzero & \toyone \\
Distinctness & $\mathrm{id}^{A}_1$ & \toyzero & \toyzero & \toyone & \toyone & \toyzero & \toyzero & \toyzero & \toyzero & \toyzero & \toyzero & \toyzero & \toyzero \\
Distinctness & $\mathrm{id}^{A}_2$ & \toyzero & \toyone & \toyzero & \toyone & \toyzero & \toyzero & \toyzero & \toyzero & \toyzero & \toyzero & \toyzero & \toyzero \\
\bottomrule
\end{tabular}

\medskip
\begin{tabular}{@{}l@{\qquad}l@{\qquad}l@{}}
\toprule
Selected vectors & diagnostic & outcome \\
\midrule
$\{v_\alpha,v_1^{21},v_2^{12}\}$ & every coordinate product is $0$ & accept \\
$\{v_\alpha,v_1^{11},v_2^{12}\}$ & product $1$ on $q_{f_1}$ & absent-edge rejection \\
$\{v_\alpha,v_1^{21},v_2^{21}\}$ & $C^{10}$-product $0$, $C^{01}$-product $1$ & endpoint disagreement \\
$\{\delta^1,v_1^{21},v_2^{12}\}$ & product $1$ on $q_\alpha$ & padding rejection \\
\bottomrule
\end{tabular}
\endgroup
\caption{The complete $9\times12$ toy $\kOV$ matrix, where $v_i^{uw}:=v_{f_i,u,w}$. Rows $C^{10},C^{01}$ compare the host indices at $p_2$ encoded by the edge vectors for $f_1$ and $f_2$, using $f_1$ as the reference edge. They abbreviate $C^{10}_{p_2,f_2,1},C^{01}_{p_2,f_2,1}$; at $N=2$, the encoding uses only one bit. Stars mark the orthogonal witness. Boxes display absence literals $\overline X^{\,i}_{uw}$ from \Cref{fig:toy-source}, with their evaluated bits in parentheses. The ``Distinctness'' rows distinguish vectors without affecting orthogonality.}
\label{fig:toy-ov}
\end{figure}
\FloatBarrier

\paragraph{Endpoint agreement.}\label{par:ov-consistency} To reject endpoint disagreements, introduce two coordinates $C^{10}_{a,g,t}$ and $C^{01}_{a,g,t}$ for each reference comparison $(a,g)\in\Ical(P)$ and bit position $t\in[\ell_N]$, where $\ell_N:=\lceil\log N\rceil$ as in \Cref{sec:projection-skeleton}. This gives $2|\Ical(P)|\ell_N$ comparison coordinates. For a fixed $(a,g)$, put $f=\rho(a)$. Coordinate $C^{10}_{a,g,t}$ assigns $\operatorname{bit}_t(\idx_a(f,u,w))$ to edge vectors in $U_f$ and $1-\operatorname{bit}_t(\idx_a(g,u,w))$ to edge vectors in $U_g$; $C^{01}_{a,g,t}$ complements these assignments on edge vectors. Both coordinates have entry $1$ on every padding vector and on every vector in an array other than $U_f$ or $U_g$.

When the $k$ selected vectors contain no padding, write $x,y$ for the corresponding bits encoding the host indices at $a$ in the selected vectors from $U_f$ and $U_g$. The products on $C^{10}_{a,g,t}$ and $C^{01}_{a,g,t}$ are $x(1-y)$ and $(1-x)y$, so both vanish exactly when $x=y$. Thus the reference comparisons enforce agreement of the host indices. In the source example of \Cref{fig:toy-source}, $f_1$ and $f_2$ share $p_2$, so a colored copy must use the same host vertex for $p_2$ on both edges. We therefore compare the second index of $v_1^{uw}$ with the first index of $v_2^{uw}$. At $N=2$, host indices $1,2$ are encoded as bits $0,1$, respectively.

Consider the third selection $\{v_\alpha,v_1^{21},v_2^{21}\}$ in \Cref{fig:toy-ov}. Both represented edges are present, but the edge represented by $v_1^{21}$ has host index $1$ at $p_2$, while the edge represented by $v_2^{21}$ has host index $2$, giving $x=0,y=1$. In row $C^{10}$, these edge columns contain $x=0$ and $1-y=0$; in row $C^{01}$, they contain $1-x=1$ and $y=1$. The auxiliary vector has entry $1$ on both rows, so the coordinate products are $1\cdot0\cdot0=0$ and $1\cdot1\cdot1=1$, respectively. Thus $C^{01}$ rejects this selection because the edges disagree at $p_2$.


\paragraph{Vector distinctness.}\label{par:ov-identifiers} The edge-check and comparison coordinates may give equal vectors at different positions of an array. In \Cref{fig:toy-ov}, columns $v_1^{12},v_1^{22}$ agree on all these rows, but $\mathrm{id}^{E}_1$ distinguishes them. We append position identifiers in two blocks of $\ell_n:=\lceil\log n\rceil$ coordinates, giving $2\ell_n$ additional coordinates to meet the distinctness requirement in \Cref{def:target-problems}.

Let $\mathrm{code}(j)$ encode $j-1$ in $\ell_n$ bits. Append $(\mathrm{code}(j),0^{\ell_n})$ to the $j$-th vector in each edge array and $(0^{\ell_n},\mathrm{code}(j))$ to the $j$-th vector in $U_\alpha$, including padding vectors. Distinct indices have distinct codes, so these coordinates distinguish all $n$ vectors within each array for every source input.

These identifiers preserve orthogonality: every coordinate in the first block is zero throughout $U_\alpha$, and every coordinate in the second is zero throughout each edge array. Since $k\ge2$, there is at least one edge array. Thus every choice of one vector per array has product zero in each added coordinate. When $n=1$, the blocks are empty and each array already contains just one vector.

\paragraph{Dimension.} There are $1+e(P)=k$ coordinates for edge presence and padding exclusion, $2|\Ical(P)|\ell_N$ for endpoint agreement and $2\ell_n$ for identifiers. Hence the vectors have dimension
\begin{equation}\label{eq:ov-dimension-budget}
\Dzero(P,N,n):=k+2|\Ical(P)|\ell_N+2\ell_n.
\end{equation}

\begin{proof}[Correctness and representation length]
Suppose first that the source instance $X$ is accepting, and let $\phi:V(P)\to[N]$ be an accepting map as in \Cref{def:psub}. For each pattern edge $f=(f_{\mathrm L},f_{\mathrm R})$, select $v_{f,\phi(f_{\mathrm L}),\phi(f_{\mathrm R})}$ from $U_f$; also select $v_\alpha$ from $U_\alpha$. The selected edge vectors represent exactly the edges of the colored copy specified by $\phi$. The resulting selection has product zero on $q_\alpha$ and every $q_f$, since no padding is chosen and every represented edge is present. At each shared endpoint $a$, all selected incident edges have host index $\phi(a)$, so the comparison products are zero as well. The identifier coordinates also have zero product. Thus all coordinate products vanish, and the OV instance accepts.

Conversely, suppose the constructed OV instance accepts, and choose one vector from each array whose coordinate products are all zero. Coordinate $q_\alpha$ forces the choice of $v_\alpha$, and each $q_f$ forces the choice of an edge vector $v_{f,u_f,w_f}$ with $X_{f,u_f,w_f}=1$. The comparison coordinates ensure that, at each non-isolated pattern vertex $a$, every selected edge object for an incident pattern edge encodes the same host index as the object for $\rho(a)$. Define $\phi(a)$ to be this common host index, assigning an arbitrary host index in $[N]$ to each isolated vertex. Then $\phi(f_{\mathrm L})=u_f$ and $\phi(f_{\mathrm R})=w_f$ for every $f$. Since every represented edge is present, $\phi$ satisfies \eqref{eq:psub-acceptance}, and the source instance accepts.

The size condition $N^2\le n$ ensures that each edge array, and also the auxiliary array, can be filled to size $n$. The \hyperref[par:ov-identifiers]{identifier coordinates} make the vectors within every array distinct for every source input. Moreover, $\Dzero(P,N,n)\ge2\ell_n\ge\ell_n$ gives $n\le2^{\Dzero(P,N,n)}$, so the construction satisfies the OV distinctness requirement at its native dimension.

Only edge-check entries depend on the source input, each as $\overline X_{f,u,w}$; all other entries are constants. Hence the construction is a depth-zero projection. Appending zero coordinates preserves correctness and distinctness for every $D\ge\Dzero(P,N,n)$. Finally, \Cref{eq:comparison-count} gives
\begin{equation}\label{eq:ov-dimension-asymptotic}
\Dzero(P,N,n)=O(k\log(2N)+\log(2n))=O(k\log(2n)),
\end{equation}
using $N^2\le n$ in the last step.
\end{proof}

\subsection{\texorpdfstring{Odd $\kXOR$}{Odd k-XOR}}\label{sec:inst-xor}

Assume $k$ is odd and $1+e(P)N^2\le n$. We encode the auxiliary object $\alpha$ and all $e(P)N^2$ edge objects $c_{f,u,w}$ in \eqref{eq:genuine-object-skeleton} as binary vectors in a single indexed list, then add padding vectors to reach length $n$. The target accepts exactly when $k$ distinct list indices have vector XOR zero (\Cref{def:target-problems}); equal vectors at different indices are allowed. Unlike OV, choosing $k$ indices does not itself enforce one edge choice per edge group. We use parity equations, together with odd $k$, to force the auxiliary object and exactly one edge object per edge group, then check edge presence and endpoint agreement.

The vectors form the columns of a matrix whose rows are the equations below; all unspecified matrix entries are assigned $0$. In \Cref{fig:toy-xor}, columns are labeled by their objects, and stars mark the accepting selection. Write $S\subseteq[n]$ for a set of $k$ selected indices.

\begin{figure}[ht]
\centering
\begingroup
\small
\newcommand{\toyconst}[1]{\cellcolor{black!14}\ensuremath{#1}}
\newcommand{\toyone}{\toyconst{1}}
\newcommand{\toyzero}{\ensuremath{0}}
\newcommand{\toylit}[3]{\begingroup\setlength{\fboxsep}{.3pt}\setlength{\fboxrule}{.25pt}\fbox{\ensuremath{\overline{X}^{\,#1}_{#2}\!{}_{\scriptscriptstyle(#3)}}}\endgroup}
\newcommand{\toywit}[1]{\ensuremath{\boldsymbol{#1}\!{}^\star}}
\setlength{\tabcolsep}{1.6pt}
\renewcommand{\arraystretch}{1.10}
\begin{tabular}{@{}c|l|c|cccc|cccc|ccc@{}}
\toprule
Role & row & \toywit{\alpha} & $c_1^{11}$ & $c_1^{12}$ & \toywit{c_1^{21}} & $c_1^{22}$ & $c_2^{11}$ & \toywit{c_2^{12}} & $c_2^{21}$ & $c_2^{22}$ & $\delta_1$ & $\delta_2$ & $\delta_3$ \\
\midrule
Selection & $S_{f_1}$ & \toyone & \toyone & \toyone & \toyone & \toyone & \toyzero & \toyzero & \toyzero & \toyzero & \toyzero & \toyzero & \toyzero \\
Selection & $S_{f_2}$ & \toyone & \toyzero & \toyzero & \toyzero & \toyzero & \toyone & \toyone & \toyone & \toyone & \toyzero & \toyzero & \toyzero \\
Edges & $H_{f_1}$ & \toyzero & \toylit{1}{11}{1} & \toylit{1}{12}{1} & \toylit{1}{21}{0} & \toylit{1}{22}{1} & \toyzero & \toyzero & \toyzero & \toyzero & \toyzero & \toyzero & \toyzero \\
Edges & $H_{f_2}$ & \toyzero & \toyzero & \toyzero & \toyzero & \toyzero & \toylit{2}{11}{1} & \toylit{2}{12}{0} & \toylit{2}{21}{0} & \toylit{2}{22}{1} & \toyzero & \toyzero & \toyzero \\
Agreement & $C_{p_2,f_2,1}$ & \toyzero & \toyzero & \toyone & \toyzero & \toyone & \toyzero & \toyzero & \toyone & \toyone & \toyzero & \toyzero & \toyzero \\
Padding & $A_{\mathrm{pad}}$ & \toyzero & \toyzero & \toyzero & \toyzero & \toyzero & \toyzero & \toyzero & \toyzero & \toyzero & \toyone & \toyone & \toyone \\
\bottomrule
\end{tabular}

\medskip
\begin{tabular}{@{}l@{\qquad}l@{\qquad}l@{}}
\toprule
Selected objects & unique nonzero row XOR & outcome \\
\midrule
$\{\alpha,c_1^{21},c_2^{12}\}$ & none & accept \\
$\{\alpha,c_1^{11},c_2^{12}\}$ & $H_{f_1}=1$ & one absent edge \\
$\{\alpha,c_1^{21},c_2^{21}\}$ & $C_{p_2,f_2,1}=1$ & endpoint disagreement \\
$\{\delta_1,\delta_2,\delta_3\}$ & $A_{\mathrm{pad}}=1$ & odd padding count \\
\bottomrule
\end{tabular}
\endgroup
\caption{The complete $6\times12$ toy $\kXOR$ matrix, where $c_i^{uw}:=c_{f_i,u,w}$ and $\delta_1,\delta_2,\delta_3$ denote the padding indices. Row $C_{p_2,f_2,1}$ compares the bits encoding the host indices at $p_2$ in the edge objects for $f_1=\rho(p_2)$ and $f_2$. Stars mark the accepting selection. Boxes show absence literals from \Cref{fig:toy-source}, with their evaluated bits in parentheses.}
\label{fig:toy-xor}
\end{figure}
\FloatBarrier

\paragraph{Selection and padding exclusion.}\label{par:xor-selection} Introduce one selection row $S_f$ for each $f\in E(P)$, giving $e(P)=k-1$ rows. It is $1$ on the auxiliary object and every edge object in group $f$. Selecting the auxiliary object and one edge object per group gives two ones in each selection row, hence XOR zero. Add one padding row $A_{\mathrm{pad}}$, equal to $1$ on every padding vector, so that an odd number of padding choices is rejected.

To see why these rows together force the intended selection when $k$ is odd, let $\zeta\in\bits$ indicate whether $S$ contains the auxiliary object, let $c_f$ count its selected edge objects in group $f$, and let $d$ count its selected padding objects. Zero XORs on the selection and padding rows force
\begin{equation}\label{eq:xor-selector-equation}
c_f\equiv\zeta\pmod2\quad(f\in E(P)),\qquad d\equiv0\pmod2,
\end{equation}
while
\begin{equation}\label{eq:xor-support-size}
\zeta+\sum_fc_f+d=k=e(P)+1.
\end{equation}
If $\zeta=0$, every $c_f$ and $d$ is even, so the selected set has even size, contradicting odd $k$. Thus $\zeta=1$, and every $c_f$ is a positive odd integer. The auxiliary object and at least one object from each of the $k-1$ groups already account for all $k$ indices, forcing $c_f=1$ for every $f$ and $d=0$.

In \Cref{fig:toy-xor}, the accepting selection has two ones per selection row; choosing all three padding indices instead gives XOR $1\oplus1\oplus1=1$ on $A_{\mathrm{pad}}$.

\paragraph{Edge presence.}\label{par:xor-group-guards} Introduce one edge-check row $H_f$ for each $f\in E(P)$, giving $e(P)$ rows. Set it to $\overline X_{f,u,w}$ on edge object $c_{f,u,w}$ in group $f$. Once the selection and padding rows have XOR zero, exactly one edge object $c_{f,u_f,w_f}$ is selected from this group, so
\begin{equation}\label{eq:xor-selected-guard-sum}
\bigoplus_{j\in S}H_f(j)=1-X_{f,u_f,w_f}.
\end{equation}
Thus this row has XOR zero for a present edge and XOR one for an absent edge. The second selection of \Cref{fig:toy-xor} fails on $H_{f_1}$ because $X^1_{11}=0$.

\paragraph{Endpoint agreement.}\label{par:xor-consistency} For each reference comparison $(a,g)\in\Ical(P)$ and bit position $t\in[\ell_N]$, introduce a row $C_{a,g,t}$, giving $|\Ical(P)|\ell_N$ rows. Put $f=\rho(a)$. On an edge object $c_{h,u,w}$ in either group $h\in\{f,g\}$, this row is $\operatorname{bit}_t(\idx_a(h,u,w))$, the corresponding bit encoding the host index at $a$ on the represented edge. Recall that $\operatorname{bit}_t(x)$ is the $t$-th bit of the $\ell_N$-bit encoding of $x-1$. With exactly one selected edge object per group, the row XOR is
\begin{equation}\label{eq:xor-consistency-row-sum}
\operatorname{bit}_t(\idx_a(f,u_f,w_f))\oplus\operatorname{bit}_t(\idx_a(g,u_g,w_g)),
\end{equation}
which vanishes exactly when those two bits agree. Testing all bits makes the two host indices equal, and comparing every incident edge with its reference edge makes all host indices at a shared endpoint agree.

In the third selection of \Cref{fig:toy-xor}, the edges represented by $c_1^{21}$ and $c_2^{21}$ use host indices $1$ and $2$, respectively, at $p_2$. Since these host indices are encoded as bits $0,1$, the row $C_{p_2,f_2,1}$ has XOR $0\oplus1=1$ and rejects their disagreement. In the accepting selection, the edges represented by $c_1^{21}$ and $c_2^{12}$ both use host index $1$ at $p_2$, giving XOR $0\oplus0=0$ on that row.

\paragraph{Row count.} There are $e(P)$ selection rows, $e(P)$ edge-check rows, $|\Ical(P)|\ell_N$ comparison rows and one padding row. Hence
\begin{equation}\label{eq:xor-row-budget}
\mxor(P,N):=2e(P)+|\Ical(P)|\ell_N+1=O(k\log(2N)),
\end{equation}
using \Cref{eq:comparison-count}.

\begin{proof}[Correctness and representation length]
Let $\phi:V(P)\to[N]$ be an accepting source map as in \Cref{def:psub}. Select the indices of $\alpha$ and $c_{f,\phi(f_{\mathrm L}),\phi(f_{\mathrm R})}$ for each $f\in E(P)$. These $k$ distinct indices contain no padding, and every selection row has two ones. Edge presence and endpoint agreement give XOR zero on the remaining rows, so the constructed XOR instance accepts.

Conversely, suppose $k$ distinct indices have vector XOR zero. The \hyperref[par:xor-selection]{selection and padding equations} force them to represent $\alpha$ and exactly one edge object $c_{f,u_f,w_f}$ per group, with no padding. Each row $H_f$ forces $X_{f,u_f,w_f}=1$, and the comparison rows ensure that the edge objects for all pattern edges incident to a non-isolated vertex $a$ encode the same host index at $a$. Define $\phi(a)$ to be this common host index, assigning an arbitrary host index in $[N]$ at each isolated vertex. Every selected edge has endpoints $\phi(f_{\mathrm L})=u_f$ and $\phi(f_{\mathrm R})=w_f$ and is present, so $\phi$ satisfies \eqref{eq:psub-acceptance}. Hence the source instance accepts.

The size condition $1+e(P)N^2\le n$ permits all the objects and the required padding. Only edge-check entries depend on the source input, each as $\overline X_{f,u,w}$; every other output bit is constant. Thus the construction is a depth-zero projection. Appending zero rows preserves every selected XOR, proving the result for all $m\ge\mxor(P,N)$.
\end{proof}

For even $k$, selecting $k$ padding indices gives XOR zero in every row, so this construction fails whenever that many padding indices are available. Our XOR projection therefore requires odd $k$.

\subsection{\texorpdfstring{$\kSUM$}{k-SUM}}\label{sec:inst-sum}

Assume $1+e(P)N^2\le n$. We encode the objects in \eqref{eq:genuine-object-skeleton}, together with padding to length $n$, as an indexed list of residues modulo $2^m$. The target accepts exactly when the residues at $k$ distinct indices sum to zero modulo $2^m$ (\Cref{def:target-problems}); repeated residue values are allowed. We first represent each object by a signed integer vector whose zero-sum equations enforce the required choices. We then pack its coordinates into one residue, choosing a base large enough that a modular zero sum is equivalent to an exact vector zero sum. This construction works for either parity of $k$.

We use exact integer equations and positional packing, as in Karp's Exact Cover to Knapsack reduction~\cite{Karp72}. Represent the object at index $j$ by a vector $\vecop(j)\in\Z^{R(P)}$, with $R(P)$ given by \eqref{eq:sum-vector-dimension}. We define its coordinates below and set all unspecified entries to zero. For a set $S\subseteq[n]$ of $k$ indices, write
\begin{equation}\label{eq:sum-support-vector}
y_S:=\sum_{j\in S}\vecop(j).
\end{equation}
The upper matrix of \Cref{fig:toy-sum} shows these signed vectors as columns labeled by their objects; the lower tables show their packed residues and selected sums.

\begin{figure}[ht]
\centering
\begingroup
\small
\newcommand{\toyconst}[1]{\cellcolor{black!14}\ensuremath{#1}}
\newcommand{\toyone}{\toyconst{1}}
\newcommand{\toyzero}{\ensuremath{0}}
\newcommand{\toylit}[3]{\begingroup\setlength{\fboxsep}{.25pt}\setlength{\fboxrule}{.25pt}\fbox{\ensuremath{\overline{X}^{\,#1}_{#2}\!{}_{\scriptscriptstyle(#3)}}}\endgroup}
\newcommand{\toywit}[1]{\ensuremath{\boldsymbol{#1}\!{}^\star}}
\setlength{\aboverulesep}{0pt}
\setlength{\belowrulesep}{0pt}
\setlength{\extrarowheight}{.75ex}
\setlength{\tabcolsep}{1.7pt}
\renewcommand{\arraystretch}{1.13}
\begin{tabular}{@{}c|l|c|cccc|cccc|ccc@{}}
\toprule
Role & coordinate & \toywit{\alpha} & $c_1^{11}$ & $c_1^{12}$ & \toywit{c_1^{21}} & $c_1^{22}$ & $c_2^{11}$ & \toywit{c_2^{12}} & $c_2^{21}$ & $c_2^{22}$ & $\delta_1$ & $\delta_2$ & $\delta_3$ \\
\midrule
Selection & $G_{f_1}$ & \toyconst{-1} & \toyone & \toyone & \toyone & \toyone & \toyzero & \toyzero & \toyzero & \toyzero & \toyzero & \toyzero & \toyzero \\
Selection & $G_{f_2}$ & \toyconst{-1} & \toyzero & \toyzero & \toyzero & \toyzero & \toyone & \toyone & \toyone & \toyone & \toyzero & \toyzero & \toyzero \\
Agreement & $C_{p_2,f_2}$ & \toyzero & \toyconst{-1} & \toyconst{-2} & \toyconst{-1} & \toyconst{-2} & \toyone & \toyone & \toyconst{2} & \toyconst{2} & \toyzero & \toyzero & \toyzero \\
Edges & $Q$ & \toyzero & \toylit{1}{11}{1} & \toylit{1}{12}{1} & \toylit{1}{21}{0} & \toylit{1}{22}{1} & \toylit{2}{11}{1} & \toylit{2}{12}{0} & \toylit{2}{21}{0} & \toylit{2}{22}{1} & \toyzero & \toyzero & \toyzero \\
Padding & $Q_{\mathrm{pad}}$ & \toyzero & \toyzero & \toyzero & \toyzero & \toyzero & \toyzero & \toyzero & \toyzero & \toyzero & \toyone & \toyone & \toyone \\
\bottomrule
\end{tabular}

\medskip
\setlength{\tabcolsep}{4pt}
$B=16,\quad \Phi=G_{f_1}+BG_{f_2}+B^2C_{p_2,f_2}+B^3Q+B^4Q_{\mathrm{pad}},\quad M=B^5=2^{20}.$

\medskip
\setlength{\tabcolsep}{2.7pt}
\begin{tabular}{@{}l|c|cccc|cccc|ccc@{}}
\toprule
standard residue & \toywit{\alpha} & $c_1^{11}$ & $c_1^{12}$ & \toywit{c_1^{21}} & $c_1^{22}$ & $c_2^{11}$ & \toywit{c_2^{12}} & $c_2^{21}$ & $c_2^{22}$ & $\delta_1$ & $\delta_2$ & $\delta_3$ \\
\midrule
hexadecimal & $\mathtt{FFFEF}$ & $\mathtt{00F01}$ & $\mathtt{00E01}$ & $\mathtt{FFF01}$ & $\mathtt{00E01}$ & $\mathtt{01110}$ & $\mathtt{00110}$ & $\mathtt{00210}$ & $\mathtt{01210}$ & $\mathtt{10000}$ & $\mathtt{10000}$ & $\mathtt{10000}$ \\
\bottomrule
\end{tabular}

\medskip
\setlength{\tabcolsep}{3.5pt}
\begin{tabular}{@{}l@{\qquad}l@{\qquad}l@{}}
\toprule
Selected objects & exact signed sum $y_S$ & sum modulo $M$ (hexadecimal) \\
\midrule
$\{\alpha,c_1^{21},c_2^{12}\}$ & $(0,0,0,0,0)$ & $\mathtt{00000}$ \\
$\{\alpha,c_1^{21},c_2^{21}\}$ & $(0,0,1,0,0)$ & $\mathtt{00100}$ \\
$\{\alpha,c_1^{11},c_2^{12}\}$ & $(0,0,0,1,0)$ & $\mathtt{01000}$ \\
$\{\delta_1,\delta_2,\delta_3\}$ & $(0,0,0,0,3)$ & $\mathtt{30000}$ \\
$\{c_1^{11},c_1^{12},\delta_1\}$ & $(2,0,-3,2,1)$ & $\mathtt{11D02}$ \\
\bottomrule
\end{tabular}
\endgroup
\caption{The five signed coordinates and twelve packed residues for the toy $\kSUM$ instance, with $c_i^{uw}:=c_{f_i,u,w}$ and padding indices $\delta_1,\delta_2,\delta_3$. Coordinate $C_{p_2,f_2}$ subtracts the host index at $p_2$ encoded by an edge object for $f_1=\rho(p_2)$ from the host index encoded by an edge object for $f_2$; coordinate $Q$ counts absent selected edges. Stars mark the accepting selection. For this example, take base $B=16$, let $\Phi$ be the displayed weighted sum of coordinates, and set modulus $M=B^5=2^{20}$. The lower tables give $\Phi$ modulo $M$ in hexadecimal; for $m>20$, multiply every residue by $2^{m-20}$.}
\label{fig:toy-sum}
\end{figure}
\FloatBarrier

\paragraph{Exact selection and padding exclusion.}\label{par:sum-selection} For each $f\in E(P)$, introduce a selection coordinate $G_f$, giving $e(P)$ coordinates. It is $-1$ on $\alpha$ and $1$ on each edge object in group $f$, so the intended choice of $\alpha$ and one object per group has sum $-1+1=0$ on every $G_f$. Add one padding coordinate $Q_{\mathrm{pad}}$, equal to $1$ on every padding object.

Let $\zeta\in\bits$ indicate whether $S$ contains $\alpha$, let $c_f$ count its edge objects in group $f$, and let $d$ count its padding objects. The sums on $G_f$ and $Q_{\mathrm{pad}}$ are $c_f-\zeta$ and $d$, respectively. If $y_S=0$, these give $c_f=\zeta$ and $d=0$. Since $|S|=k$ and $e(P)=k-1$, we obtain $k=\zeta+\sum_fc_f+d=k\zeta$. Thus $\zeta=1$ and $c_f=1$ for every $f$. In \Cref{fig:toy-sum}, the fourth selection has sum $3$ on $Q_{\mathrm{pad}}$, while the fifth has sum $2$ on $G_{f_1}$ and $1$ on $Q_{\mathrm{pad}}$; they are rejected as neither has vector sum zero.


\paragraph{Edge presence.} A single coordinate $Q$ checks all selected edges: set it to $\overline X_{f,u,w}$ on every edge object $c_{f,u,w}$. Once the selection equations force one object $c_{f,u_f,w_f}$ per group, the sum on $Q$ is $\sum_{f\in E(P)}(1-X_{f,u_f,w_f})$. Every summand is nonnegative, so this sum is zero exactly when all represented edges are present.

\paragraph{Endpoint agreement.} For each $(a,g)\in\Ical(P)$, introduce one comparison coordinate $C_{a,g}$, giving $|\Ical(P)|$ coordinates. Put $f=\rho(a)$ and assign $-\idx_a(f,u,w)$ on edge objects in group $f$ and $+\idx_a(g,u,w)$ on edge objects in group $g$. Here we use the integer host indices themselves. With one selected object per group, the sum on this coordinate is
\begin{equation}\label{eq:sum-consistency-coordinate}
\idx_a(g,u_g,w_g)-\idx_a(f,u_f,w_f),
\end{equation}
which is zero exactly when the two host indices agree. The reference comparisons therefore make all selected edges at a shared endpoint agree on its host index.

For the running example, $C_{p_2,f_2}$ is $-w$ on $c_1^{uw}$ and $+u$ on $c_2^{uw}$, comparing the second index of the first edge with the first index of the second. In the accepting selection $\{\alpha,c_1^{21},c_2^{12}\}$, this gives $0-1+1=0$. In the second selection, replacing $c_2^{12}$ by $c_2^{21}$ gives $0-1+2=1$, detecting the disagreement at $p_2$.

\paragraph{Number of coordinates.} There are $e(P)$ selection coordinates, $|\Ical(P)|$ comparison coordinates, one edge-check coordinate and one padding coordinate. Thus
\begin{equation}\label{eq:sum-vector-dimension}
R(P):=e(P)+|\Ical(P)|+2\le3e(P)+2=O(k),
\end{equation}
by \Cref{eq:comparison-count}.

\paragraph{Packing into a residue.} Fix an order of the $R=R(P)$ coordinates. Every coordinate of a single vector has absolute value at most $N$, so every coordinate of $y_S$ has absolute value at most $kN$. We give successive coordinates powers of a base $B>2kN$. This separation will ensure that a nonzero vector sum cannot pack to zero modulo $B^R$. Set
\begin{equation}\label{eq:sum-packing-parameters}
q:=\lceil\log(2kN+1)\rceil,\qquad B:=2^q,\qquad \msum(P,N):=qR.
\end{equation}
Then $B>2kN$ and $2^{\msum(P,N)}=B^R$. For an integer vector $z=(z_0,\ldots,z_{R-1})\in\Z^R$, define the integer packing and its standard residues by
\begin{equation}\label{eq:sum-packing-map}
\Phi(z):=\sum_{i=0}^{R-1}z_iB^i,\qquad \widehat a_j:=\Phi(\vecop(j))\bmod B^R\in[0,B^R).
\end{equation}
At the native width $\msum(P,N)$, the target list consists of the residues $\widehat a_j$. At any requested width $m\ge\msum(P,N)$, use
\begin{equation}\label{eq:sum-top-bit-scaling}
a_j:=2^{m-\msum(P,N)}\widehat a_j\in[0,2^m).
\end{equation}

Multiplication by this power of two appends $m-\msum(P,N)$ zero bits at the least significant end (the right) of each native encoding. The following lemma proves that packing and scaling preserve exactly the vector zero sums.

\begin{lemma}[Packing preserves zero sums]\label{lem:sum-packed-supports-exact}
For every set $S\subseteq[n]$ of $k$ indices and every $m\ge\msum(P,N)$,
\begin{equation}\label{eq:sum-packed-support-equivalence}
\sum_{j\in S}a_j\equiv0\pmod{2^m}\quad\Longleftrightarrow\quad y_S=0.
\end{equation}
\end{lemma}
\begin{proof}
We have $\|y_S\|_\infty\le kN<B/2$. For any integer vector $y$ in this range, suppose $B^R\mid\Phi(y)$. Reducing modulo $B$ gives $B\mid y_0$, and $|y_0|<B/2$ forces $y_0=0$. Dividing by $B$ leaves $B^{R-1}\mid\sum_{i=1}^{R-1}y_iB^{i-1}$, so the same argument forces $y_1=0$. Repeating gives $y=0$. The converse is immediate, so
\begin{equation}\label{eq:sum-faithful-packing}
\Phi(y)\equiv0\pmod{B^R}\quad\Longleftrightarrow\quad y=0\qquad(\|y\|_\infty\le kN).
\end{equation}
Linearity gives $\sum_{j\in S}\widehat a_j\equiv\Phi(y_S)\pmod{B^R}$, proving the claim at the native width. For any integer $Z$, cancelling the common power of two gives
\begin{equation}\label{eq:sum-scaling-divisibility}
2^m\mid2^{m-\msum(P,N)}Z\quad\Longleftrightarrow\quad 2^{\msum(P,N)}\mid Z.
\end{equation}
Applying this with $Z=\sum_{j\in S}\widehat a_j$ proves the claim for every larger width as well.
\end{proof}


\begin{proof}[Correctness and representation length]
Let $\phi:V(P)\to[N]$ be an accepting source map as in \Cref{def:psub}. Select the indices of $\alpha$ and $c_{f,\phi(f_{\mathrm L}),\phi(f_{\mathrm R})}$ for every $f\in E(P)$. These $k$ distinct indices contain no padding and have $y_S=0$ by the selection equations, edge presence and endpoint agreement. Hence \Cref{lem:sum-packed-supports-exact} makes the constructed SUM instance accept.

Conversely, let $S$ be a set of $k$ indices whose residues sum to zero modulo $2^m$. By \Cref{lem:sum-packed-supports-exact}, $y_S=0$. The \hyperref[par:sum-selection]{exact selection equations} force $\alpha$ and one edge object $c_{f,u_f,w_f}$ per group, with no padding. Coordinate $Q$ forces $X_{f,u_f,w_f}=1$ for every $f$, and the comparison coordinates force endpoint agreement. Define $\phi(a)$ to be the common host index at each non-isolated vertex, assigning arbitrary host indices in $[N]$ to isolated vertices. Then $\phi(f_{\mathrm L})=u_f$ and $\phi(f_{\mathrm R})=w_f$ for every $f$, so edge presence gives \eqref{eq:psub-acceptance}, and the source accepts.

The size condition $1+e(P)N^2\le n$ permits the auxiliary object, all edge objects and the required padding. The auxiliary and padding encodings are constant. Each edge object's vector depends only on its own source bit $X_{f,u,w}$ through coordinate $Q$. Its packed and scaled residue therefore takes one of two fixed $m$-bit values, according as $X_{f,u,w}=0$ or $1$. Each output bit is consequently a constant, $X_{f,u,w}$ or $\overline X_{f,u,w}$, so the construction is a depth-zero projection, including any carries in the packing. Finally, \Cref{eq:sum-vector-dimension,eq:sum-packing-parameters} give $\msum(P,N)=O(k\log(2kN))$.
\end{proof}

\paragraph{Longer representations.}\label{par:parameter-extension} The representations can be extended to any larger dimension or bit length: append zero coordinates for $D\ge\Dzero(P,N,n)$ in OV, append zero rows for $m\ge\mxor(P,N)$ in XOR, and use \Cref{eq:sum-top-bit-scaling,eq:sum-scaling-divisibility} for $m\ge\msum(P,N)$ in SUM.
\FloatBarrier

\subsection{Source-to-target transfer}\label{sec:lower-bound-transfer}\label{sec:master-projections}

This section demonstrates how a circuit lower bound for colored subgraph isomorphism yields lower bounds for the three targets. A larger source host strengthens the resulting bound, so we want to choose $N$ as large as the target's arrays or list can accommodate. \Cref{thm:A} states the conditions for the objects to fit in the target input and the required representation lengths; \Cref{lem:host-substitutions} then chooses the hosts and converts a source bound $N^a$ into a power of $n$ for OV or $\frac{n}{k}$ for XOR/SUM.

\paragraph{Notation for bounds.} We generally write source lower bounds with exponent rate $\gamma$ and target lower bounds with rate $\beta$ or $c$; \Cref{lem:host-substitutions} explains how these rates are related. The universal constants $\COV,\Cxor,\Csum$ in \Cref{thm:A} specify the dimension, row-count, and bit-width requirements for the projections.

We use $\Dzero(P,N,n)$, $\mxor(P,N)$ and $\msum(P,N)$ for the representation lengths required by the three constructions, as defined in \Cref{eq:ov-dimension-budget,eq:xor-row-budget,eq:sum-packing-parameters}, respectively.

\begin{theorem}[Source-to-target projections]\label{thm:A}
There are universal constants $\COV,\Cxor,\Csum>0$ such that the following holds. Let $k\ge2$, let $P$ be a simple loopless pattern with $e(P)=k-1$. Then the following depth-zero projections exist.
\begin{enumerate}[label=(\roman*),itemsep=2pt]
\item If $N^2\le n$ and $D\ge\Dzero(P,N,n)$, then $P\text{-}\SUB_N\le_{\proj}\kOV_{n,D,k}$, where $\Dzero(P,N,n)\le\COV k\log(2n)$. Moreover, $\Dzero(P,N,n)\ge\lceil\log n\rceil$, so $D$ automatically satisfies $n\le2^D$.
\item If $k$ is odd, $1+e(P)N^2\le n$, and $m\ge\mxor(P,N)$, then $P\text{-}\SUB_N\le_{\proj}\kXOR_{n,m,k}$, where $\mxor(P,N)\le\Cxor k\log(2N)$.
\item If $1+e(P)N^2\le n$ and $m\ge\msum(P,N)$, then $P\text{-}\SUB_N\le_{\proj}\kSUM_{n,m,k}$, where $\msum(P,N)\le\Csum k\log(2kN)$.
\end{enumerate}
\end{theorem}

\begin{proof}
\Cref{sec:inst-ov,sec:inst-xor,sec:inst-sum} establish the respective projections and bounds on representation length. We choose $\COV,\Cxor,\Csum$ once, with $\Cxor\ge1$, to satisfy both \Cref{thm:A} and the estimates of the required representation lengths in \Cref{lem:host-substitutions}.
\end{proof}

The source size $N$ remains a free parameter, allowing OV to use larger source instances than the single-list targets.

\begin{lemma}[Choosing source sizes and transferring exponents]\label{lem:host-substitutions}\label{tab:intro-target-interface}
Let $k\ge2$, $n\ge16k$, and let $P$ be a simple loopless pattern with $e(P)=k-1$. Choose the source problem parameters
\begin{equation}\label{eq:fixed-k-host-substitutions}
N_{\mathrm{OV}}:=\lfloor\sqrt n\rfloor,
\qquad
N_{\oplus,\Sigma}:=\left\lfloor\sqrt{\frac{n-1}{k-1}}\right\rfloor.
\end{equation}
These choices ensure that the auxiliary and edge objects fit in the target arrays or list.

Substituting these source sizes into the representation bounds of \Cref{thm:A} gives an OV dimension of at most $\COV k\log n$ and an XOR bit length of at most $\Cxor k\log(n/k)$. The SUM bit length is at most $\Csum k\log n$, improving to $\Csum k\log(n/k)$ when $n\ge k^3$.

For a target problem $Q$, write $N_Q$ for its source parameter and put $B_Q:=n$ for OV and $B_Q:=n/k$ for XOR or SUM. If $Q$ meets these representation requirements, with odd $k$ for XOR, then every source bound $\size_d(P\text{-}\SUB_{N_Q})\ge N_Q^a$ with $a>0$ gives
\begin{equation}\label{eq:host-transfer}
\size_d(Q)\ge(B_Q/4)^{a/2}\ge B_Q^{a/4}.
\end{equation}
The same conclusion holds separately for $\Sigma_3$ or $\Pi_3$ circuits.
\end{lemma}

\begin{proof}
The choices in \eqref{eq:fixed-k-host-substitutions} satisfy $N_{\mathrm{OV}}^2\le n$ and $1+(k-1)N_{\oplus,\Sigma}^2\le n$, so all objects fit in their target arrays or list. Since $(n-1)/(k-1)\ge n/k\ge16$, the inequality $\lfloor x\rfloor\ge x/2$ for $x\ge2$ gives $N_Q\ge\sqrt{B_Q}/2$. Also $B_Q\ge16$, so $\log N_Q\ge\tfrac12\log B_Q-1\ge\tfrac14\log B_Q$.

For the representation lengths, $\log(2n)\le2\log n$ and $N_{\oplus,\Sigma}\le\sqrt{2n/k}$. Under $n\ge16k$, the latter gives $\log(2N_{\oplus,\Sigma})\le\tfrac32+\tfrac12\log(n/k)\le\log(n/k)$ and $\log(2kN_{\oplus,\Sigma})\le\log n$. If $n\ge k^3$, then $kN_{\oplus,\Sigma}\le\sqrt{2kn}\le\sqrt2\,n/k$, so $\log(2kN_{\oplus,\Sigma})\le\tfrac32+\log(n/k)\le2\log(n/k)$. The bounds in \Cref{thm:A}, with the universal construction constants chosen large enough to cover these factors of two, give all the stated representation requirements.

Apply \Cref{thm:A} and \hyperref[par:projection-monotonicity]{projection monotonicity} to obtain $\size_d(Q)\ge N_Q^a\ge(\sqrt{B_Q}/2)^a=(B_Q/4)^{a/2}$. Since $B_Q\ge16$, this is at least $B_Q^{a/4}$, proving \eqref{eq:host-transfer}. The projections preserve orientation, so the same proof applies separately to $\Sigma_3$ and $\Pi_3$.
\end{proof}



\section{Lower Bounds}\label{sec:lower-bounds}\label{sec:growing-k}

We now combine the projections of \Cref{thm:A} with three source lower bounds. For fixed $k$, the LRR theorem gives an exponent linear in $k$. For growing $k$, Beame's clique bound applies at every fixed depth. At depth three, two projections of the regular core give an exponent linear in $k$ for both gate orientations. Each application checks that the hosts in \Cref{lem:host-substitutions} meet the source theorem's requirements.

\subsection{Shared source patterns}\label{sec:shared-patterns}

The common projections (\Cref{thm:A}) apply to any source pattern with exactly $k-1$ edges. To obtain lower bounds, we choose patterns for which colored subgraph isomorphism is hard. The shared family below starts with a bipartite $4$-regular expander core and adds at most three disjoint edges to reach the required edge count. For a core graph $F$, let $J_t$ be a matching of $t\ge0$ edges on new vertices, with $J_0$ empty. We write $P=F\mathbin{\dot\cup}J_t$ for their disjoint union.

Keep the endpoint orders $(f_{\mathrm L},f_{\mathrm R})$ chosen for $F$ when forming $P$, and choose those for the matching edges arbitrarily. Fix once and for all a universal threshold $k_0\ge600$ large enough for the following construction to work for every $k\ge k_0$; \Cref{app:expander-patterns} shows that such a finite threshold exists.

\begin{lemma}[Padded bipartite expander patterns]\label{lem:expander-pattern-family}
For every integer $k\ge k_0$, there is a simple loopless graph
\begin{equation}\label{eq:expander-family-decomposition}
P^{\mathrm{exp}}_k=F^{\mathrm{exp}}_k\mathbin{\dot\cup}J_{t_k}
\end{equation}
with exactly $k-1$ edges, no isolated vertices, and maximum degree $4$, where $F^{\mathrm{exp}}_k$ is a simple connected bipartite $4$-regular expander and $0\le t_k<4$. The core has
\begin{equation}\label{eq:expander-core-size}
e(F^{\mathrm{exp}}_k)>k-5,
\qquad
v(F^{\mathrm{exp}}_k)=2\left\lfloor\frac{k-1}{4}\right\rfloor.
\end{equation}
\end{lemma}

\paragraph{Construction sketch.} Put $r:=\lfloor(k-1)/4\rfloor$. Start with a graph $Q$ on exactly $r$ vertices, with no loops or repeated edges. Each vertex is joined to four of the other $r-1$ vertices. Make a left copy $a_{\mathrm L}$ and a right copy $a_{\mathrm R}$ of each vertex $a$ of $Q$. For every edge $\{a,b\}$ of $Q$, add the edges $\{a_{\mathrm L},b_{\mathrm R}\}$ and $\{b_{\mathrm L},a_{\mathrm R}\}$ to the new graph, as in \Cref{fig:double-cover-edge-rule}. The resulting core, called the \emph{bipartite double cover}, has $r$ vertices on each side and four neighbors per vertex.

\begin{figure}[htbp]
\centering
\begin{tikzpicture}[x=1cm,y=1cm,>=Latex,every node/.style={font=\small},vertex/.style={circle,draw,minimum size=6mm,inner sep=1pt}]
\node at (0,1.2) {one edge of $Q$};
\node[vertex] (a) at (-0.65,0) {$a$};
\node[vertex] (b) at (0.65,0) {$b$};
\draw[thick] (a) -- (b);
\draw[->,thick] (1.35,0) -- (2.2,0);
\node at (3,1.2) {left copy};
\node at (5.2,1.2) {right copy};
\node[vertex] (aL) at (3,0.6) {$a_{\mathrm L}$};
\node[vertex] (bL) at (3,-0.6) {$b_{\mathrm L}$};
\node[vertex] (aR) at (5.2,0.6) {$a_{\mathrm R}$};
\node[vertex] (bR) at (5.2,-0.6) {$b_{\mathrm R}$};
\draw[thick,blue!70!black] (aL) -- (bR);
\draw[thick] (bL) -- (aR);
\end{tikzpicture}
\caption{Each edge $\{a,b\}$ of the base graph gives two edges, $\{a_{\mathrm L},b_{\mathrm R}\}$ and $\{b_{\mathrm L},a_{\mathrm R}\}$. This rule is applied to every edge of $Q$.}
\label{fig:double-cover-edge-rule}
\end{figure}

We choose $Q$ using Friedman's spectral theorem for simple regular graphs~\cite[Thm.~1]{Bordenave20} so that the core is connected and has constant expansion: isolating a set containing at most half the core vertices requires cutting a number of edges proportional to the size of that set. The core has $4r$ edges, so adding at most three disjoint edges on at most six new vertices brings the total to exactly $k-1$. The choice of $Q$ and the verification of these properties are in \Cref{app:expander-patterns}.

Fix this family for the rest of the paper. It supports fixed-pattern hardness in \Cref{sec:fixed-pattern-hardness}, the depth-three projections in \Cref{thm:D}, and the conjecture behind \Cref{thm:E}. For growing $k$ at arbitrary fixed depth, \Cref{thm:C} uses a clique of growing degree padded by a matching. These are the two family branches of \Cref{fig:intro-architecture}.

\subsection{Fixed \texorpdfstring{$k$}{k}}\label{sec:fixed-k}

For fixed $k$, the expander pattern is fixed as well, so the LRR theorem applies once the host size $N$ is sufficiently large. Its exponent grows linearly with $k$ along our pattern family.

\subsubsection{Fixed-pattern hardness}\label{sec:fixed-pattern-hardness}

Let $\kappa(P)$ denote the colored parameter $\kappa_{\mathrm{col}}(P)$ of~\cite[Def.~2.12(iii)]{LRR17}, following the notation of~\cite[Footnote~1]{Rosenthal19}. This parameter determines the exponent in the following lower bound.

\begin{theorem}[Li--Razborov--Rossman fixed-pattern lower bound~\cite{LRR17}]\label{thm:lrr}
For every fixed simple loopless pattern $P$ with at least one edge and no isolated vertices, every fixed depth $d$, and every constant $\eta>0$ satisfying $\kappa(P)-\eta>2$, there is a threshold $N_0(P,d,\eta)$ such that, for every $N\ge N_0(P,d,\eta)$,
\begin{equation}\label{eq:lrr-quantified}
\size_d(P\text{-}\SUB_N)\ge N^{\kappa(P)-\eta}.
\end{equation}
\end{theorem}

\Cref{app:lrr-assembly} derives this statement in our gate-count convention. The theorem applies separately to each fixed pattern and does not control the threshold $N_0(P,d,\eta)$ as $P$ grows.

The spectral form of LRR's expansion bound (\Cref{eq:kappa-spectral}) shows that a regular graph of bounded degree and constant spectral gap has $\kappa$ at least a fixed positive fraction of its number of vertices. Our core has $\Theta(k)$ vertices, so its $\kappa$ is linear in $k$. We apply this bound to the core. Deleting the added matching components from the full pattern recovers the core, so minor monotonicity gives $\kappa(P^{\mathrm{exp}}_k)\ge\kappa(F^{\mathrm{exp}}_k)$. The calculation in \Cref{app:expander-patterns} makes the constant explicit:
\begin{equation}\label{eq:expander-pattern-kappa}
\kappa(P^{\mathrm{exp}}_k)\ge c_\kappa(k-1),
\qquad c_\kappa:=1/128.
\end{equation}
The threshold $k_0\ge600$ fixed in \Cref{lem:expander-pattern-family} ensures $c_\kappa(k_0-1)>4$, which will guarantee the hypothesis $\kappa(P)-\eta>2$ in the following application.

\begin{corollary}[Fixed-$k$ source bound with a universal exponent constant]\label{cor:fixed-source-rate}
Fix any constant $\gamma$ with $0<\gamma<c_\kappa$. For every fixed depth $d$ and every fixed integer $k\ge k_0$, there is a threshold $N_0(\gamma,d,k)$ such that
\begin{equation}\label{eq:fixed-source-rate}
\size_d(P^{\mathrm{exp}}_k\text{-}\SUB_N)\ge N^{\gamma(k-1)}
\end{equation}
for every $N\ge N_0(\gamma,d,k)$.
\end{corollary}

\begin{proof}
For fixed $d$ and $k\ge k_0$, choose $\eta:=(c_\kappa-\gamma)(k-1)/2>0$. By \Cref{lem:expander-pattern-family}, the fixed pattern $P^{\mathrm{exp}}_k$ is simple and loopless, has $k-1\ge1$ edges, and has no isolated vertices. Moreover, \Cref{eq:expander-pattern-kappa} gives
\begin{equation*}
\kappa(P^{\mathrm{exp}}_k)-\eta
\ge\frac{c_\kappa+\gamma}{2}(k-1)
\ge\tfrac12c_\kappa(k_0-1)>2.
\end{equation*}
Thus \Cref{thm:lrr} applies. Since $(c_\kappa+\gamma)/2>\gamma$, it yields \Cref{eq:fixed-source-rate} for all sufficiently large $N$.
\end{proof}

\subsubsection{Transfer to the three targets}\label{sec:fixed-target-bounds}

Applying \Cref{lem:host-substitutions} to \Cref{cor:fixed-source-rate} gives base $n$ for OV and $n/k$ for XOR and SUM. The same constant in the target exponent works at every fixed depth and for every sufficiently large fixed $k$; only the minimum required $n$ may depend on $(d,k)$.

\begin{theorem}[Unconditional lower bounds for fixed $k$]\label{thm:B}
Let $k_0\in\N$ be the universal threshold of \Cref{lem:expander-pattern-family}, and let $\COV,\Cxor,\Csum$ be the universal construction constants fixed in \Cref{thm:A,lem:host-substitutions}. There is a universal constant $\beta>0$, independent of $k$, $d$, and $n$, such that for every fixed depth $d$ and every fixed integer $k\ge k_0$ there is a threshold $n_0(d,k)$ for which the following hold for every $n\ge n_0(d,k)$:
\begin{align}
\size_d(\kOV_{n,D,k})&\ge n^{\beta(k-1)} &&\text{for every }D\ge\COV k\log n,\label{eq:B-ov}\\
\size_d(\kXOR_{n,m,k})&\ge(n/k)^{\beta(k-1)} &&\text{for odd }k\text{ and every }m\ge\Cxor k\log(n/k),\label{eq:B-xor}\\
\size_d(\kSUM_{n,m,k})&\ge(n/k)^{\beta(k-1)} &&\text{for every }m\ge\Csum k\log(n/k).\label{eq:B-sum}
\end{align}
\end{theorem}

\begin{proof}
We use \Cref{cor:fixed-source-rate} to obtain a source exponent linear in $k$, then \Cref{lem:host-substitutions} to transfer it to the targets. The source coefficient $\gamma$ can be any constant below $c_\kappa=1/128$ from \Cref{eq:expander-pattern-kappa}. The square-root source sizes halve this coefficient asymptotically, so choose universal constants
\begin{equation}\label{eq:B-rate-slack}
0<2\beta<\gamma<c_\kappa.
\end{equation}
For each fixed $(d,k)$, \Cref{cor:fixed-source-rate} gives
\begin{equation}\label{eq:B-source-rate}
\size_d(P^{\mathrm{exp}}_k\text{-}\SUB_N)\ge N^{\gamma(k-1)}
\qquad(N\ge N_0(\gamma,d,k)).
\end{equation}
For each target $Q$, take $N_Q$ and $B_Q$ as in \Cref{lem:host-substitutions}: $B_Q=n$ for OV and $B_Q=n/k$ for XOR and SUM. Require $n\ge16k$ to apply the lemma's host-size and representation estimates, and $n\ge k^3$ to obtain the sharper SUM bit length $\Csum k\log(n/k)$. Thus the theorem's representation hypotheses suffice for the projections. Once $N_Q\ge N_0(\gamma,d,k)$, the source bound applies, and the first inequality in \eqref{eq:host-transfer} gives
\begin{equation*}
\log\left(\size_d(Q)\right) \ge\frac{\gamma(k-1)}2\log(B_Q/4)=\gamma(k-1)\bigl(\tfrac12\log B_Q-1\bigr).
\end{equation*}
Here $\log(B_Q/4)=\log B_Q-2$, since all logarithms are base two. To absorb the resulting loss into the slack between $\gamma/2$ and $\beta$, require
\begin{equation*}
\log(n/k)\ge\frac{\gamma}{\gamma/2-\beta}.
\end{equation*}
Since $B_Q\ge n/k$, this ensures $(\gamma/2-\beta)\log B_Q\ge\gamma$. Rearranging gives $\frac\gamma2\log B_Q-\gamma\ge\beta\log B_Q$, so
\begin{equation*}
\log\left(\size_d(Q)\right)\ge(k-1)\bigl(\tfrac\gamma2\log B_Q-\gamma\bigr)\ge\beta(k-1)\log B_Q.
\end{equation*}
Exponentiating yields $\size_d(Q)\ge B_Q^{\beta(k-1)}$, as required.

Finally, for fixed $(d,k)$, all these requirements hold for every sufficiently large $n$: both source sizes in \Cref{eq:fixed-k-host-substitutions} and $\log(n/k)$ tend to infinity. Choose $n_0(d,k)$ so that they hold simultaneously for every $n\ge n_0(d,k)$.
\end{proof}

\Cref{thm:parameter-fixed} extends these bounds to other representation lengths. We call $D=\Theta(k\log n)$ for OV and $m=\Theta(k\log(n/k))$ for XOR and SUM the \emph{standard representation scales}, with constant factors large enough to meet the lower-bound hypotheses. \Cref{app:upper-bounds} compares the lower bounds at these scales with elementary $\Sigma_3$ upper bounds.

\subsection{Growing \texorpdfstring{$k$}{k}: a clique bound at every fixed depth}\label{sec:growing-clique}

For growing $k$, we need a source theorem that explicitly relates the threshold of host size $N$ to the pattern size. Beame's small-clique bound provides this relation. A clique with at most $k-1$ edges has only $O(\sqrt{k})$ vertices, limiting the dependence on $k$ in Beame's exponent to order $\sqrt{k}$. Its additional requirement of at most $\log N$ clique vertices limits the exponent further once $k$ reaches order $\log^2 n$.

\begin{lemma}[Colored-clique consequence of Beame's lower bound]\label{lem:beame}
For every fixed depth $d$, there are a constant $a_d>0$ and integers $r_d,N_d\ge1$ such that, whenever $N\ge N_d$ and $r_d\le r\le\log N$,
\begin{equation}\label{eq:beame-small-clique}
\size_d(K_r\text{-}\SUB_N)\ge N^{a_dr}.
\end{equation}
Here $K_r\text{-}\SUB_N$ is the colored source of \Cref{def:psub} with the $r$-vertex clique as its pattern and $N$ host vertices per color class. The upper limit $r\le\log N$ is independent of depth; only $a_d,r_d,N_d$ may depend on $d$.
\end{lemma}

\Cref{app:beame-import} derives this colored bound from Beame's circuit lower bound~\cite[Cor.~3.8]{Beame90}, checking size conventions and giving the projection from ordinary to colored clique.

The explicit range $r\le\log N$ lets us choose $r$ as a function of $N$. Rossman's stronger exponent is proved for every fixed number $r$ of clique vertices, and his paper discusses a growing-$r$ extension as a possibility~\cite[Thm.~1.2; Sec.~6]{Rossman08}.

For $r\ge2$ and $t:=k-1-\binom r2\ge0$, use the padded clique $P^{\mathrm{clq}}_{k,r}:=K_r\mathbin{\dot\cup}J_t$ of \Cref{def:clique-pattern}, where $J_t$ is a matching of $t$ edges. Fixing all matching-edge variables of $P^{\mathrm{clq}}_{k,r}\text{-}\SUB_N$ to $1$ yields $K_r\text{-}\SUB_N$ (\Cref{lem:clique-source-chain}). Thus the lower bound in \Cref{lem:beame} transfers to a source with exactly the $k-1$ edges required by \Cref{thm:A}.

\begin{theorem}[Growing $k$ at every fixed depth]\label{thm:C}
Fix a constant $\varepsilon>0$. For every fixed depth $d$, the following hold along every sequence $k=k(n)$ satisfying $k\to\infty$ and $k\le n^{1-\varepsilon}$, for all sufficiently large $n$. Define
\begin{equation}\label{eq:C-rhos}
\rho_{\mathrm{OV}}:=\min\{\sqrt{k-1},\log n\},
\qquad
\rho:=\min\{\sqrt{k-1},\log(n/k)\}.
\end{equation}
Then
\begin{align}
\size_d(\kOV_{n,D,k})&\ge\exp\bigl(\Omega_d(\rho_{\mathrm{OV}}\log n)\bigr) &&\text{for every }D\ge\COV k\log n,\label{eq:C-ov}\\
\size_d(\kXOR_{n,m,k})&\ge\exp\bigl(\Omega_d(\rho\log(n/k))\bigr) &&\text{for odd }k\text{ and every }m\ge\Cxor k\log(n/k),\label{eq:C-xor}\\
\size_d(\kSUM_{n,m,k})&\ge\exp\bigl(\Omega_d(\rho\log(n/k))\bigr) &&\text{for every }m\ge\Csum k\log n.\label{eq:C-sum}
\end{align}
The implied constants in the three lower bounds depend only on $d$.
\end{theorem}

\begin{proof}
Fix $d$ and the constants $a_d,r_d,N_d$ from \Cref{lem:beame}. Let $Q$ be one of $\kOV_{n,D,k}$, $\kXOR_{n,m,k}$, or $\kSUM_{n,m,k}$. Use the hosts $N_Q$ and bases $B_Q$ of \Cref{lem:host-substitutions}, with $B_Q=n$ for OV and $B_Q=n/k$ for XOR and SUM. Let $\rho_Q$ be $\rho_{\mathrm{OV}}$ for OV and $\rho$ otherwise. Since $n/k\ge n^\varepsilon\to\infty$, eventually $n\ge16k$, as required in \Cref{lem:host-substitutions}. Then its proof gives
\begin{equation*}
N_Q\ge\frac{\sqrt{B_Q}}2,
\qquad
\log N_Q\ge\frac12\log B_Q-1\ge\frac14\log B_Q,
\end{equation*}
where the last inequality uses $B_Q\ge n/k\ge16$. In particular, $N_Q\to\infty$, so eventually $N_Q\ge N_d$, meeting the host-size threshold in \Cref{lem:beame}.

We choose the clique size to meet both the requirement of at most $k-1$ pattern edges and Beame's restriction $r\le\log N_Q$. Set
\begin{equation}\label{eq:C-clique-choice}
r:=\left\lfloor\rho_Q/4\right\rfloor.
\end{equation}
Both $\sqrt{k-1}$ and $\log B_Q\ge\varepsilon\log n$ tend to infinity, so their minimum $\rho_Q$ does too. For sufficiently large $n$, we have $\rho_Q\ge\max\{8,4r_d\}$. Then $r\ge r_d$ and $r\ge2$, meeting the lower thresholds for \Cref{lem:beame} and the padded-clique construction above. Also, $r=\left\lfloor\rho_Q/4\right\rfloor\ge\rho_Q/8$. For Beame's upper restriction, the definition of $\rho_Q$ and the displayed host estimate give
\begin{equation*}
r\le\frac{\rho_Q}4\le\frac14\log B_Q\le\log N_Q.
\end{equation*}

The same definition gives $r\le\sqrt{k-1}$, so the clique has at most $k-1$ edges. We can therefore add $t:=k-1-\binom r2\ge0$ disjoint matching edges to obtain the pattern $P:=P^{\mathrm{clq}}_{k,r}$ with exactly $k-1$ edges, as described before the theorem. Padding preserves the colored-clique lower bound by \Cref{lem:clique-source-chain}, and the verified host and clique-size conditions permit \Cref{lem:beame}. Together they give
\begin{equation}\label{eq:C-source-lower}
\size_d(P\text{-}\SUB_{N_Q})\ge\size_d(K_r\text{-}\SUB_{N_Q})\ge N_Q^{a_dr}.
\end{equation}

The theorem's representation lengths meet the requirements of \Cref{lem:host-substitutions}, using the general SUM width bound $\Csum k\log n$. Applying \eqref{eq:host-transfer} to \eqref{eq:C-source-lower} with $a=a_dr$ and using $r\ge\rho_Q/8$ yields
\begin{equation*}
\log\left(\size_d(Q)\right)\ge\frac{a_dr}{4}\log B_Q\ge\frac{a_d}{32}\rho_Q\log B_Q.
\end{equation*}
Substituting the definitions of $B_Q$ and $\rho_Q$ gives the three asserted bounds, with constants depending only on $d$.
\end{proof}

\paragraph{The range of the Beame-based bound.}\label{par:clique-plateau} Under the representation requirements of \Cref{thm:C}, the inequalities $\log(n/k)\ge\varepsilon\log n$ and $\sqrt{k-1}\ge\sqrt{k}/2$ give
\begin{equation}\label{eq:C-common-epsilon-form}
\log\size_d\ge\Omega_{d,\varepsilon}\bigl(\min\{\sqrt{k},\log n\}\cdot\log n\bigr).
\end{equation}
If $k\to\infty$ and $k=n^{o(1)}$, then $\log(n/k)=(1-o(1))\log n$, so
\begin{equation}\label{eq:C-subpolynomial-form}
\size_d\ge n^{\Omega_d(\min\{\sqrt{k},\log n\})}.
\end{equation}
In this subpolynomial range, once $k\ge\log^2 n$, Beame's $r\le\log N$ limit fixes the bound in \Cref{thm:C} at $n^{\Omega_d(\log n)}$.
Chukhin gives an ordinary-clique gate lower bound of $N^{\Omega_{d,\delta}(r)}$ for fixed $d\ge2$ and $0<\delta<1/(d-1)$ in the wider range $3\le r\le N^{1/(d-1)-\delta}$~\cite[Thm.~1.2]{Chukhin26}. The same ordinary-to-colored projection~\eqref{eq:ordinary-to-colored-clique} applies throughout that range.

\subsection{Growing \texorpdfstring{$k$}{k}: linear exponents at depth three}\label{sec:growing-depth-three}

To obtain a top-OR lower bound with a linear exponent at depth three, we reduce \emph{iterated intersection} to detecting a colored copy of the connected core $F^{\mathrm{exp}}_k$ from \Cref{lem:expander-pattern-family}. Given a list $(A_u,B_u)_{u\in[a]}$ of pairs of subsets of $[b]$, iterated intersection asks whether $A_u\cap B_u\ne\varnothing$ for every $u$. Both projections use the following table formulation.

\begin{definition}[Table function]\label{def:table-function}
Let $a,b,t$ be integers with $a,b\ge1$ and $t\ge2$. An input $x$ to $H_{a,b,t}$ is a table with $a$ rows and $b$ columns, with bits $x_{u,j,1},\ldots,x_{u,j,t}$ in entry $(u,j)$. The function $H_{a,b,t}(x)$ equals $1$ exactly when every row contains an entry whose $t$ bits are all $1$. A \emph{failed row} contains a zero in every entry, so $\neg H_{a,b,t}(x)=1$ exactly when some row fails. For $t=2$, taking $x_{u,j,1}$ and $x_{u,j,2}$ to be the membership bits of $j$ in $A_u$ and $B_u$, respectively, gives iterated intersection.
\end{definition}

A top-OR circuit is a disjunction of CNFs. To lower-bound its size, we simplify the clauses by restricting the table inputs. Independently in each entry, we pick one bit uniformly at random to leave free and fix the others to $1$. A clause with many positive literals is likely to be forced true by one of those fixed bits. Each such clause comes from a bottom OR gate, which may feed several CNFs. If the circuit has few gates, some choice of which bits are restricted to $1$ forces all such clauses to be true at once, so they can be discarded.

After the restriction, an entry is all $1$ exactly when its free bit is $1$. Setting one free bit per row to $1$ and all other free bits to $0$ therefore gives a minimal accepting assignment because changing any of the chosen $1$ bits to $0$ leaves its row with no all-$1$ entry. When the remaining clauses have few positive literals, each CNF can accept only a small fraction of these assignments, so the circuit needs many CNFs to ensure correctness.

For the top-AND lower bound, we allow more than two bits per entry and change the host encoding so that a colored copy of the same core exists exactly when some row contains a zero in every entry. The projected circuit thus computes the complement of the table function. Complementing this circuit gives a top-OR circuit of the same size for the table function, so the same counting argument applies.

\begin{theorem}[Depth three]\label{thm:D}
For every constant $\gamma$ with $0<\gamma<1/2$, there are $k_D(\gamma)\ge k_0$ and $C_D(\gamma)>0$ such that
\begin{equation}\label{eq:D-source}
\size_3(P^{\mathrm{exp}}_k\text{-}\SUB_N)\ge\size_3(F^{\mathrm{exp}}_k\text{-}\SUB_N)>N^{\gamma k}
\end{equation}
for every $k\ge k_D(\gamma)$ and every $N\ge k^{C_D(\gamma)}$. For every constant $c$ with $0<c<\gamma/2$, there are $k_1(\gamma,c)\in\N$ and $C'(\gamma,c)>1$ such that, for every $k\ge k_1(\gamma,c)$ and every $n\ge k^{C'(\gamma,c)}$,
\begin{align}
\size_3(\kOV_{n,D,k})&\ge n^{ck} &&\text{for every }D\ge\COV k\log n,\label{eq:D-ov}\\
\size_3(\kXOR_{n,m,k})&\ge(n/k)^{ck} &&\text{for odd }k\text{ and every }m\ge\Cxor k\log(n/k),\label{eq:D-xor}\\
\size_3(\kSUM_{n,m,k})&\ge(n/k)^{ck} &&\text{for every }m\ge\Csum k\log(n/k).\label{eq:D-sum}
\end{align}
Every fixed $0<c<1/4$ is therefore attainable.
\end{theorem}

To count how many minimal accepting assignments a single CNF can cover, we represent each assignment by its \emph{support}, the set of input variables it sets to $1$. For a set $U$ of input variables, let $\mathbf{1}_U$ denote the assignment with support $U$. A \emph{minterm} of a Boolean function $f$ is a set $U$ such that $f(\mathbf{1}_U)=1$ and $f(\mathbf{1}_W)=0$ for every proper subset $W\subsetneq U$.

\begin{lemma}[Counting minterms of a CNF]\label{lem:cnf-count-minterms}
If every clause of a CNF contains at most $w$ distinct positive literals, then the CNF has at most $w^s$ minterms of size $s$, for each $s\ge1$.
\end{lemma}

\begin{proof}

Fix $s$, and let $M(W)$ count all size-$s$ minterms containing $W$. If $|W|=s$, then $M(W)\le1$. Suppose $|W|<s$ and $M(W)>0$. Then $W$ is a proper subset of a minterm, so the CNF rejects $\mathbf{1}_W$. Choose one clause $C$ falsified by $\mathbf{1}_W$. Every minterm counted by $M(W)$ must satisfy this same clause, no matter how many other clauses $\mathbf{1}_W$ falsifies.

Let $C^+$ be the set of variables occurring positively in $C$. Since $C$ is false under $\mathbf{1}_W$, every variable in $C^+$ lies outside $W$, and the variable of each negative literal lies in $W$. The negative literals therefore stay false on every superset of $W$. Thus each minterm counted by $M(W)$ contains some $x\in C^+$, giving
\[
M(W)\le\sum_{x\in C^+} M(W\cup\{x\}).
\]
A minterm containing several variables of $C^+$ can contribute to several summands. There are at most $w$ summands, each indexed by a support of size $|W|+1$. Starting from the bound $1$ at size $s$ and inducting down to size $0$ gives $M(\varnothing)\le w^s$, as claimed.
\end{proof}

Both projections need a top-OR lower bound for the table function, with two bits per entry for iterated intersection and longer entries for the failed-row projection. If the circuit is small, we can leave one bit free in each entry while forcing clauses with many positive literals to be true. \Cref{lem:cnf-count-minterms} then bounds how many choices of one entry per row each CNF can cover.

For two-bit entries, H{\aa}stad, Jukna, and Pudl{\'a}k proved the lower bound $\min\{2^w,(b/w)^a\}$ for top-AND circuits computing $\neg H_{a,b,2}$~\cite[Eq.~(4.1)]{HJP95}. Choudhury and Sreenivasaiah state its top-OR form for $H_{a,b,2}$~\cite[Prop.~27]{ChoudhurySreenivasaiah24}. We adapt the restriction argument to arbitrary $t\ge2$.

\begin{lemma}[Top-OR lower bound for the table function]\label{lem:table-top-or}
If a depth-three circuit with an OR output gate ($\Sigma_3$) computes the table function $H_{a,b,t}$ of \Cref{def:table-function} with $S$ non-input gates, then for every integer $1\le w\le b$,
\begin{equation}\label{eq:D-table-bound}
S+2abt\ \ge\ \min\bigl\{t^w,(b/w)^a\bigr\}.
\end{equation}
\end{lemma}

\begin{proof}

If $S+2abt\ge t^w$, the inequality already holds. Suppose $S+2abt<t^w$, so our goal is $S+2abt\ge(b/w)^a$. The disjunction contains at most $S+2abt$ CNFs by \eqref{eq:sigma3-normal-form}, counting literals wired directly to the output as one-clause CNFs. It suffices to find a restriction under which the table has $b^a$ minterms of size $a$ and each CNF covers at most $w^a$ of them. To obtain the latter bound from \Cref{lem:cnf-count-minterms}, we first force every clause with more than $w$ distinct positive literals to be true, so it can be removed from every CNF that uses it.

In each entry $(u,j)$, leave one uniformly chosen bit $y_{u,j}$ free and fix the other $t-1$ bits to $1$, choosing independently across entries. Therefore, entry $(u,j)$ is all-$1$ exactly when $y_{u,j}=1$. Thus row $u$ becomes $\bigvee_{j=1}^b y_{u,j}$. Any of its $b$ columns can be its unique all-$1$ entry. Choosing one free bit per row to set to $1$ and setting all other free bits to $0$ gives the $b^a$ minterms of size $a$.

Consider a clause with $p>w$ distinct positive variables. It can survive only if all $p$ remain free, since fixing any one of them to $1$ forces the clause true. This is impossible if any two of them lie in the same entry. If they lie in distinct entries, each is free with probability $1/t$, independently, so all $p$ are free with probability $t^{-p}$. Since $p\ge w+1$, the clause survives with probability at most $t^{-(w+1)}$ in either case. Every such clause comes from a bottom OR gate, and a gate shared by several CNFs need only be counted once. By a union bound over at most $S$ gates, the probability that any such clause survives is at most $S t^{-(w+1)}<1/t<1$, using $S<t^w$. We can therefore fix free positions that force all these clauses to be true at once.

For each of the $b^a$ minterms $U$, some restricted CNF accepts $\mathbf{1}_U$. That CNF rejects every proper subset of $U$ because the whole disjunction rejects it, so $U$ is also a minterm of the CNF. After the restriction, each surviving clause has at most $w$ distinct positive literals, so each CNF covers at most $w^a$ of these minterms by \Cref{lem:cnf-count-minterms}. This gives $b^a\le(S+2abt)w^a$, proving \eqref{eq:D-table-bound}.
\end{proof}

\begin{proof}[Proof of \Cref{thm:D}]
Fix $k\ge k_0$ and let $F:=F^{\mathrm{exp}}_k$ be the connected $4$-regular core. It has $r_F:=v(F)=2\lfloor(k-1)/4\rfloor=k/2+O(1)$ vertices by \eqref{eq:expander-core-size}.

\paragraph{Encoding iterated intersection for top OR.} Give each vertex $v$ of $F$ a row of the table $H_{r_F,N,2}$. Write $A_v,B_v\subseteq[N]$ for the sets whose membership bits in column $j$ are $x_{v,j,1}$ and $x_{v,j,2}$, respectively. An accepting witness chooses an element $j(v)\in A_v\cap B_v$ for each row $v$. We want every such choice to give a colored copy on the host vertices $(v,j(v))$. In particular, choosing an element for row $v$ must leave us free to choose any element of $A_{v'}\cap B_{v'}$ for a neighboring row $v'$. We achieve this by assigning each pattern edge to test one membership bit at one endpoint.

Each vertex needs two tests, one for $A_v$ and one for $B_v$. Since $F$ is connected and every vertex has even degree, it has an Euler tour, a closed walk that traverses every edge exactly once. Draw an arrow $v\to v'$ on each pattern edge $f$ that the tour traverses from $v$ to $v'$. The arrow leaves $v$ (its tail) and enters $v'$ (its head). It assigns $f$ to test row $v$. At each vertex the tour enters as often as it leaves, so exactly two of the four incident edges leave that vertex. Number these outgoing edges $1$ and $2$.

Recall from \Cref{def:psub} that for a pattern edge $f=(f_{\mathrm L},f_{\mathrm R})$ and host indices $u,w\in[N]$, the input bit $X_{f,u,w}$ equals $1$ exactly when the undirected host edge between $(f_{\mathrm L},u)$ and $(f_{\mathrm R},w)$ is present. We may list the endpoints in either order. For an arrow $v\to v'$, choose $(f_{\mathrm L},f_{\mathrm R})=(v,v')$, so the first host index belongs to the row tested by $f$. If $f$ is the $i$th edge leaving $v$, set
\[
X_{f,j,j'}=x_{v,j,i}\qquad(j,j'\in[N]).
\]
For a fixed $j$ at $v$, every choice of $j'$ at $v'$ gives the same value for this edge. Under a map $j:V(F)\to[N]$, the two edges leaving $v$ therefore have host-edge bits $x_{v,j(v),1}$ and $x_{v,j(v),2}$. Both are $1$ exactly when $j(v)\in A_v\cap B_v$. Since each pattern edge leaves exactly one vertex, the map specifies a colored copy exactly when this condition holds in every row. This proves $H_{r_F,N,2}\le_{\proj}F\text{-}\SUB_N$. Substituting these bits into a top-OR circuit for the core gives a top-OR circuit for the table with at most as many gates.

\paragraph{Encoding a failed row for top AND.} Fix integers $q\ge1$ and $\ell\ge2$ with $q\ell\le N$, and write $H=H_{q,r_F,\ell}$. Start with any top-AND circuit for $F\text{-}\SUB_N$. To apply the top-OR lower bound for $H$ in \Cref{lem:table-top-or}, we seek a projection $\neg H\le_{\proj}F\text{-}\SUB_N$. Substituting this projection into the core circuit would give a top-AND circuit for $\neg H$. By De Morgan's laws, swapping AND and OR gates and complementing the leaves would then give a top-OR circuit of the same size for $H$. We must therefore build a host graph that contains a colored copy exactly when the table has a failed row.

A witness for $\neg H$ chooses a row $u$ and, for each column $v$, a position $i(v)$ with $x_{u,v,i(v)}=0$. For $F\text{-}\SUB_N$, a candidate copy chooses one host index $j(v)\in[N]$ in each color $v$, specifying the host vertex $(v,j(v))$. We index the columns of the input table by the $r_F$ colors. In each color, we label one host index for each $(u,i)\in[q]\times[\ell]$, using $q\ell$ of the $N$ indices. Choosing the index labeled $(u,i)$ in color $v$ proposes that row $u$ fails and bit $i$ in entry $(u,v)$ is $0$.

Set every host-edge bit to $0$ if either endpoint uses an index without a label. For the labeled indices, keep the Euler-tour orientation from the top-OR projection. Each host-edge bit belongs to the edge group of a unique pattern edge $f$. If the tour directs $f$ from $v$ to $v'$, then $v$ supplies the zero test for that group. For endpoint indices labeled $(u,i)$ at $v$ and $(u',i')$ at $v'$, set
\[
X_{f,(u,i),(u',i')}=
\begin{cases}
0 & \text{if }u\ne u',\\
\neg x_{u,v,i} & \text{if }u=u'.
\end{cases}
\]
A copy needs four present host edges at each $v$, and every present edge has labeled endpoints. Thus every chosen index carries a label $(u,i)$. The row labels $u, u'$ agree along each edge of the copy and hence throughout the copy by connectedness of $F$. In that common row $u$, the edges leaving each pattern vertex $v$ require its chosen position $i(v)$ to satisfy $x_{u,v,i(v)}=0$. Conversely, given a failed row $u$, choose a zero position $i(v)$ in each column $v$. Selecting the index labeled $(u,i(v))$ in each color $v$ gives a copy because every required edge is present.

Since every host-edge bit is a constant or a table literal, this proves $\neg H\le_{\proj}F\text{-}\SUB_N$. Projection monotonicity and the size-preserving De Morgan transformation in both directions give
\[
\size_{\Pi_3}(F\text{-}\SUB_N)\ge\size_{\Pi_3}(\neg H)=\size_{\Sigma_3}(H).
\]

\paragraph{The source exponent.} Let $S_{\lor}$ and $S_{\land}$ be the minimum sizes of depth-three top-OR and top-AND circuits for $F\text{-}\SUB_N$, respectively. In the failed-row table $H_{q,r_F,\ell}$, the $q$ rows share the $N$ host indices in each color, with $\ell$ indices per row. For the first table $H_{r_F,N,2}$, use this same $q$ as the cutoff $w$ on the number of distinct positive literals per clause in \Cref{lem:table-top-or}, so that $b/w=N/q\ge\ell$. For the failed-row table, write $w=h$ for the cutoff, with $1\le h\le r_F$ since this table has $r_F$ columns. The two projections then give
\begin{align*}
S_{\lor}+4r_FN &\ge \min\{2^q,(N/q)^{r_F}\},\\
S_{\land}+2qr_F\ell &\ge \min\{\ell^h,(r_F/h)^q\}.
\end{align*}
We aim for $\ell^h$ as a lower bound on both minima. With two bits per entry, the second minimum would be at most $2^{r_F}$ for every host size, so we take the largest entry length that fits: $\ell=\lfloor N/q\rfloor$. The term $(N/q)^{r_F}$ is already at least $\ell^h$. We want the exponent $h$ of $\ell$ to be close to its maximum $r_F=k/2+O(1)$. At $h=r_F$ the competing term $(r_F/h)^q$ is only $1$, so we keep $h$ a fixed fraction below $r_F$ and choose $q$ large enough for both $(r_F/h)^q$ and $2^q$ to reach $N^k$. For a fixed $0<\varepsilon<1$, set
\[
h=\lfloor(1-\varepsilon)r_F\rfloor,\qquad q=\lceil k\log N/\varepsilon\rceil,\qquad \ell=\lfloor N/q\rfloor.
\]
Fix $C>1$ and suppose $N\ge k^C$. Since $q=O_\varepsilon(k\log N)$ and $k\le N^{1/C}$, the ratio $N/q$ is at least a constant multiple of $N^{1-1/C}/\log N$. This tends to infinity uniformly for $N\ge k^C$ as $k\to\infty$. Thus $\ell\ge2$ and $\ell\ge N/(2q)$ for all sufficiently large $k$. Also $1\le h$ because $r_F=k/2+O(1)$.

The choice of $h$ gives $h\ge(1-\varepsilon)r_F-1$ while keeping $r_F/h\ge1/(1-\varepsilon)$. Our choice of $q$ now gives $2^q\ge N^k$ and $(r_F/h)^q\ge(1-\varepsilon)^{-q}\ge2^{\varepsilon q}\ge N^k$. Since $\ell\le N$ and $h\le r_F\le k$, these terms are at least $\ell^h$, as required. The additive terms on the left are $4r_FN$ and $2qr_F\ell\le2r_FN$. Subtracting them and taking the minimum of $S_{\lor}$ and $S_{\land}$ as in \eqref{eq:depth-three-size} gives $\size_3(F\text{-}\SUB_N)\ge\ell^h-4r_FN$.

To compare $\ell^h$ with $N^{\gamma k}$, we estimate the number of factors $h$ and the size of each factor $\ell$ separately. Since $r_F=k/2+O(1)$, rounding $(1-\varepsilon)r_F$ down gives $h/k=(1-\varepsilon)/2+O(1/k)$. The lower bound on $N/q$ and the inequality $\ell\ge N/(2q)$ give $\log\ell\ge(1-1/C)\log N-\log\log N-O_\varepsilon(1)$. Consequently, uniformly for $N\ge k^C$,
\[
\frac{\log(\ell^h)}{k\log N}
=\frac{h}{k}\frac{\log\ell}{\log N}
\ge\frac{(1-\varepsilon)(1-1/C)}{2}-o(1).
\]
Given $0<\gamma<1/2$, choose $C=C_D(\gamma)>1/(1-2\gamma)$, so $(1-1/C)/2>\gamma$. Then choose $\varepsilon>0$ small enough that the constant on the right still exceeds $\gamma$, and let $2\eta>0$ be the gap. For all sufficiently large $k$, the uniform error is smaller than $\eta$, so $\ell^h\ge N^{(\gamma+\eta)k}$ for every $N\ge k^C$. Increasing the threshold on $k$ if necessary gives both $4r_FN\le N^{\gamma k}$, since $r_F\le k\le N^{1/C}$, and $N^{\eta k}>2$. Thus even after subtracting $4r_FN$ the core bound exceeds $N^{\gamma k}$. For example, $C=2$ and $\varepsilon=1/4$ give a constant of $3/16$ in the display, which suffices for $\gamma=1/8$.

\paragraph{Transfer to the targets.} The padded pattern $P^{\mathrm{exp}}_k$ has exactly the $k-1$ edges required by \Cref{lem:host-substitutions}. Fixing its matching-edge variables to $1$ lets every colored core copy extend to a padded copy, and every padded copy contains a core copy. The resulting projection \eqref{eq:expander-core-padding-projection} transfers the core bound to the padded pattern, proving \eqref{eq:D-source}.

For each target $Q$, write $B_Q=n$ for OV and $B_Q=n/k$ for XOR and SUM. When $n\ge16k$, the host in \Cref{lem:host-substitutions} has size $N_Q\ge\sqrt{B_Q}/2$. Take $C'=2C_D(\gamma)+3$. For $n\ge k^{C'}$ and large $k$, we have $n\ge16k$ and $B_Q\ge n/k\ge k^{2C_D(\gamma)+2}$, so $N_Q\ge k^{C_D(\gamma)}$ as required. We also have $n\ge k^3$, which gives the stated SUM bit length. Under the hypotheses of the theorem for $Q$, \eqref{eq:host-transfer} gives
\[
\size_3(Q)\ge (B_Q/4)^{\gamma k/2}=B_Q^{(\gamma/2-\gamma/\log B_Q)k}.
\]
Since $B_Q\ge n/k\to\infty$, every fixed $c<\gamma/2$ gives the asserted $B_Q^{ck}$ lower bound for sufficiently large $k$. Choosing $\gamma$ between $2c$ and $1/2$ gives every $0<c<1/4$.
\end{proof}

The choice $C=2$ and $\gamma=1/8$ gives explicit ranges for the target bounds \eqref{eq:D-ov}--\eqref{eq:D-sum}. The hosts in \eqref{eq:fixed-k-host-substitutions} have size at least $k^2$ when $n\ge k^4$ for OV and when $n\ge k^5$ for XOR and SUM. Thus for sufficiently large $k$ the lower bounds are $n^{\Omega(k)}$ and $(n/k)^{\Omega(k)}$, respectively.

When $k\to\infty$ and $k=n^{o(1)}$, all three lower bounds become $n^{\Omega(k)}$. At the standard representation lengths, the top-OR constructions in \Cref{tab:sigma3-upper-bounds} therefore give matching upper bounds for unrestricted depth three up to constants in the exponent.

\paragraph{Comparison with PARITY.}\label{par:parity-comparison}
For odd $k\ge3$ and $n\ge4k$, a PARITY projection gives the tight depth-two bound $\exp(\Theta(k^2\log(n/k)))$ for XOR at the standard representation length $m=\Theta(k\log(n/k))$ with a sufficiently large constant factor (\Cref{cor:xor-depth-two-tight}). At $k=\Theta(\log^2 n)$, this improves on the subgraph bound at depth two, while the subgraph bounds are larger at depth three in both orientations and at every fixed $d\ge4$ (\Cref{tab:parity-comparison}). The comparison at other growing witness sizes depends on both $k$ and $d$.


\subsection{A Pattern-Uniform LRR Conjecture}\label{sec:outlook}

For fixed depths $d\ge4$, a stronger source theorem would give growing-$k$ bounds with exponent linear in $k$ through the same projections. We state this hypothesis and derive its consequences for the three targets.


The fixed-pattern theorem (\Cref{thm:lrr}) allows the minimum host size $N_0(P,d,\eta)$ to depend arbitrarily on $P$. For the padded expander family of \Cref{lem:expander-pattern-family}, we ask for an exponent linear in $k$ once the host size is polynomial in $k$, with constants depending only on the fixed depth.

\begin{conjecture}[Pattern-Uniform LRR, depth-by-depth form]\label{conj:dbd}
For the padded expander family $(P^{\mathrm{exp}}_k)_{k\ge k_0}$, for every fixed depth $d$, there are constants $\gamma_d>0$ and $c'_d>0$ such that
\begin{equation}\label{eq:uniform-lrr-conjecture}
\size_d(P^{\mathrm{exp}}_k\text{-}\SUB_N)\ge N^{\gamma_d k}
\end{equation}
for every $k\ge k_0$ and every $N\ge k^{c'_d}$.
\end{conjecture}

\Cref{thm:D} proves the conjectured bound at depth three with any fixed $0<\gamma<1/2$ and host threshold $N\ge k^{C_D(\gamma)}$, after increasing $k_0$ to $k_D(\gamma)$. Since every depth-two circuit also has depth at most three, the same bound holds at depth two. The cases $d\ge4$ remain open.

The switching estimate in~\cite{LRR17} does not directly give this polynomial threshold. With the restriction parameter $\delta>0$ in \Cref{eq:lrr-switching-expression} fixed independently of $k$, making this bound small at circuit size $N^{\gamma k}$ requires $\log N=\Omega_{d,\delta,\gamma}(k)$ (see \Cref{app:lrr-uniformity}). This observation concerns the published estimate and leaves open refinements of the method. The published range $d=o(\log N/\log\log N)$ for each fixed pattern~\cite[p.~938]{LRR17} concerns growing depth and does not give the required uniformity over patterns.

\paragraph{A stronger rate.}\label{par:universal-rate-strengthening}
Requiring $\gamma_d\ge\gamma_0>0$ for one universal constant $\gamma_0$ would make the constants hidden by $\Omega_d$ below independent of depth, as in \Cref{thm:B}. The stated depth-by-depth form already suffices for the following consequence.

\begin{corollary}[Conditional lower bounds at every fixed depth]\label{thm:E}
Assume \Cref{conj:dbd}. For every fixed depth $d$, there are integers $k_{E}(d)\ge k_0$ and constants $C_{E}(d)>1$ such that, for every $k\ge k_{E}(d)$ and every $n\ge k^{C_{E}(d)}$,
\begin{align}
\size_d(\kOV_{n,D,k})&\ge\exp\bigl(\Omega_d(k\log n)\bigr) &&\text{for every }D\ge\COV k\log n,\label{eq:E-ov}\\
\size_d(\kXOR_{n,m,k})&\ge\exp\bigl(\Omega_d(k\log n)\bigr) &&\text{for odd }k\text{ and every }m\ge\Cxor k\log(n/k),\label{eq:E-xor}\\
\size_d(\kSUM_{n,m,k})&\ge\exp\bigl(\Omega_d(k\log n)\bigr) &&\text{for every }m\ge\Csum k\log(n/k).\label{eq:E-sum}
\end{align}
Consequently, the conclusion applies along every sequence $k=k(n)\to\infty$ with $k=n^{o(1)}$.
\end{corollary}

\begin{proof}
Fix $d$ and the conjectural constants, and choose
\begin{equation}\label{eq:E-threshold-choice}
C_E(d)\ge\max\{2c'_d+4,4\}.
\end{equation}
Take $k_E(d):=k_0$. For $k\ge k_0\ge600$ and $n\ge k^{C_E(d)}$, we have $n\ge k^4\ge\max\{16k,k^3\}$, so \Cref{lem:host-substitutions} applies with the sharper SUM representation bound for $n\ge k^3$. Both host sizes in \Cref{eq:fixed-k-host-substitutions} are at least
\begin{equation*}
\frac12\sqrt{\frac nk}\ge\frac12 k^{c'_d+3/2}\ge k^{c'_d},
\end{equation*}
where the first inequality uses $C_E(d)\ge2c'_d+4$. Thus the conjecture applies at both hosts.

The projections of \Cref{thm:A} and the transfer bound \eqref{eq:host-transfer} give exponent $\gamma_dk/4$ in bases $n$ for OV and $n/k$ for XOR and SUM, with odd $k$ for XOR. The range $n\ge k^4$ also gives $\log k\le\tfrac14\log n$, hence $\log(n/k)=\log n-\log k\ge\tfrac34\log n$. Converting the XOR and SUM bounds to base $n$ therefore gives all three bounds in the stated form.
\end{proof}

\section{Acknowledgment and AI Disclosure}\label{sec:acknowledgment}

\ifanon\else
  I am deeply grateful to my advisor, Professor Prashant Nalini Vasudevan, who encouraged me to study average-case planted $\kXOR$ lower bounds in $\AC^0$ and with whom I had many formative discussions throughout 2023. His guidance and support have been the foundation of this work. I also thank Professor Alexander Golovnev for a discussion in the summer of 2024 on possible routes to $\AC^0$ lower bounds for the average-case planted problem, and for suggesting candidate source problems such as the Coin Problem; that exchange prompted the broader search that led to colored subgraph isomorphism.
\fi

GPT-5.5 pointed us to Beame's lower bound~\cite{Beame90} as a candidate source for growing $k$ at every fixed depth. We verified its applicability and developed the resulting argument in \Cref{lem:beame,thm:C}. We take responsibility for the correctness of the final statements and proofs.


\bibliographystyle{alpha}
\bibliography{refs}

\appendix
\addtocontents{toc}{\protect\setcounter{tocdepth}{1}}
\crefalias{section}{appendix}
\crefalias{subsection}{appendix}
\section{Source Functions and Cited Lower Bounds}\label{app:source-formalities}

We check the source and circuit-size conventions in the cited LRR bound, explain the host-size limitation of its switching estimate, and derive the colored-clique lower bound from Beame's circuit theorem. As in \Cref{sec:structured-source}, $N$ is the source host size; $n$ counts vectors per OV array or entries in an XOR or SUM list.

\subsection{Applying the fixed-pattern LRR bound}\label{app:lrr-assembly}\label{app:source-equivalence}

We first identify our source with LRR's structured colored problem, then check the distribution and circuit conventions needed to apply their lower bound. The source in \Cref{def:psub} differs from their colored subgraph problem only in its input-bit names.

\begin{proposition}[Coordinate-level equivalence]\label{prop:psub-equivalence}
For every simple loopless graph $P$ with a fixed endpoint order $f=(f_{\mathrm L},f_{\mathrm R})$ for each edge and every $N\ge1$, the function $P\text{-}\SUB_N$ of \Cref{def:psub} coincides, after a bijection of Boolean input coordinates, with $\mathrm{SUBGRAPH}_{\mathrm{col},N}(P)$ of~\cite[Sec.~5, p.~960]{LRR17}.
\end{proposition}

\begin{proof}
The structured host in~\cite[Sec.~5, p.~960]{LRR17} is $V(P)\times[N]$ with first-coordinate coloring, as in~\cite[Def.~2.8(v)]{LRR17} at vertex exponent $\alpha\equiv1$. Its edge coordinates are
\begin{equation}\label{eq:lrr-host-coordinate}
\{(a,u),(b,w)\}\quad\text{for }\{a,b\}\in E(P),\ u,w\in[N].
\end{equation}
Simplicity and the fixed endpoint orders make
\begin{equation}\label{eq:coordinate-bijection}
\{(f_{\mathrm L},u),(f_{\mathrm R},w)\}\longmapsto X_{f,u,w}
\end{equation}
a bijection. A color-respecting copy selects $(a,\phi(a))$ in every class. Different colors give distinct host vertices even when their indices agree. Its required edges are present exactly when \Cref{eq:psub-acceptance} holds; isolated colors may be assigned arbitrary indices.
\end{proof}

We identify LRR's hard distribution on our source bits and account for any additional input and negation gates in their size measure.

\begin{proof}[Proof of \Cref{thm:lrr}]
Fix $P,d$ and $\eta>0$ with $\kappa(P)-\eta>2$ as in the theorem. Our assumption of no isolated vertices avoids an unused boundary case at $\alpha\equiv1$. Connectedness is unnecessary: LRR's parameters apply to disconnected patterns~\cite[Rem.~2.7]{LRR17}.

By~\cite[Cor.~4.2]{LRR17}, there is an edge weighting $\beta:E(P)\to[0,2]$ with $(\1,\beta)\in\theta_{\mathrm{col}}(P)$ such that
\begin{equation}\label{eq:optimal-colored-pair}
\kappa_{\1,\beta}(P)=\kappa_{\mathrm{col}}(P)=\kappa(P).
\end{equation}
These parameters and $G_{\1,\beta}(N)$ are those of~\cite[Defs.~2.8 and~2.12]{LRR17}. The displays of their Theorems~3.11--3.12 write $G_{\alpha,\beta}(P)$; the definition and proofs identify its argument as the host parameter $N$. At vertex exponent one, its host is
\begin{equation}\label{eq:lrr-random-host}
\{(a,u):a\in V(P),\ 1\le u\le\lfloor N^1\rfloor\}=V(P)\times[N],
\end{equation}
with exactly the coordinates of \Cref{eq:lrr-host-coordinate}. By \Cref{prop:psub-equivalence}, a circuit computing our source is correct on every input from this distribution and therefore meets the correctness requirement of the cited average-case theorem.

Their type-I model uses unbounded-fan-in AND and OR gates with NOT gates~\cite[Sec.~3.2, p.~948]{LRR17}, so we apply their Theorem~3.12. Allowing depth $d+2$ covers conventions that count input or leaf-negation layers; the bound holds at every fixed depth. Their size measure counts gates~\cite[Footnote~4, p.~940]{LRR17}, without specifying whether inputs, constants or input negations count. For $S:=\size_d(P\text{-}\SUB_N)$ and $M_{\mathrm{src}}=e(P)N^2$ as in \Cref{sec:projection-skeleton}, our circuit uses at most $S+2M_{\mathrm{src}}+2$ gates under each convention: one variable and one shared negation node per source bit, and two constants.

If this additive quantity is at least $N^{\kappa(P)+1}$, subtracting $2M_{\mathrm{src}}+2=O_P(N^2)$ gives $S\ge N^{\kappa(P)-\eta}$ for all sufficiently large $N$, since $\kappa(P)>2$. Otherwise the converted circuit has polynomial size in $N$, as required by their type-I model. Theorem~3.12 and \Cref{eq:optimal-colored-pair} give
\begin{equation}\label{eq:lrr-asymptotic-bound}
S+2M_{\mathrm{src}}+2\ge N^{\kappa(P)-o_{P,d}(1)}.
\end{equation}
For each fixed $P$ and $d$, the exponent loss $o_{P,d}(1)$ tends to zero as $N\to\infty$, at a rate that may depend on $P$ and $d$. Choose $N_0(P,d,\eta)$ large enough that this loss is at most $\eta/2$. Since $\kappa(P)-\eta>2$, increasing the threshold if necessary gives
\begin{equation*}
S\ge N^{\kappa(P)-\eta/2}-2e(P)N^2-2\ge N^{\kappa(P)-\eta},
\end{equation*}
because the convention term has exponent $2$ and the desired exponent retains $\eta/2$ slack. This proves \Cref{eq:lrr-quantified} in our gate-count convention.
\end{proof}


\subsection{Host size in the cited switching estimate}\label{app:lrr-uniformity}

The threshold in \Cref{thm:lrr} may depend arbitrarily on the fixed pattern $P^{\mathrm{exp}}_k$. In particular, it gives no bound $N_0(P^{\mathrm{exp}}_k,d,\eta)\le k^{O_d(1)}$ of the kind needed in \Cref{conj:dbd}. The obstruction to obtaining this threshold directly from the published switching estimate is in its union bound.

LRR's depth-reduction argument applies the switching lemma at each gate then takes a union bound. The resulting bad-event bound is~\cite[Eq.~(11), p.~951; Appendix~C, p.~970]{LRR17}
\begin{equation}\label{eq:lrr-switching-expression}
B_{d,\delta}(S,N):=S\bigl(5N^{-\delta/d}\delta\log N\bigr)^{\delta\log N},
\end{equation}
where $S$ is the circuit size, $d$ is its depth, and $\delta>0$ is a restriction parameter. For each fixed pattern, their analysis takes $S=N^{O(1)}$ and fixes $d$ and $\delta$ before letting $N\to\infty$.

At the desired size scale $S=N^{\gamma k}$, with $d,\delta,\gamma>0$ fixed independently of $k$, taking base-two logarithms gives
\begin{equation}\label{eq:lrr-switching-log}
\log B_{d,\delta}(N^{\gamma k},N)
=
\gamma k\log N
-
\Theta_{d,\delta}\bigl((\log N)^2\bigr).
\end{equation}
At a polynomial host size $N=k^C$ with fixed $C>0$, the positive term has order $k\log k$, whereas the negative term has order $(\log k)^2$. Thus the decay supplied by the switching estimate cannot offset the union bound over $N^{\gamma k}$ gates, and the displayed probability bound becomes vacuous. More generally, making this bound tend to zero as $k,N\to\infty$ requires $\log N=\Omega_{d,\delta,\gamma}(k)$, so this estimate requires at least exponential host size in $k$.

This calculation concerns only the published estimate with $\delta$ fixed independently of $k$. It does not show that the actual bad-event probability is large, and it leaves open refinements of the method that might prove \Cref{conj:dbd}. Other pattern-dependent choices in LRR's probabilistic argument would also require uniform control.

\subsection{From Beame's circuit bound to colored cliques}\label{app:beame-import}

For integers $r\ge2$ and $N\ge1$, let $K_r\text{-}\CLIQUE_N$ be the Boolean function on variables $y_{\{u,w\}}$, indexed by unordered pairs of distinct $u,w\in[N]$, that accepts exactly when the represented graph contains an $r$-clique.

\begin{proof}[Proof of \Cref{lem:beame}]
\emph{Ordinary clique.} Beame's circuit lower bound~\cite[Cor.~3.8, p.~18]{Beame90} gives size $N^{\Omega_d(r)}$ at each fixed depth $d$, uniformly for sufficiently large $N$ and $r$ with $r\le\log N$. Here $N$ counts vertices and there are $M:=\binom N2$ input variables, as in~\cite[Sec.~2]{Beame90}. This circuit bound is also recorded in~\cite[Sec.~1]{Rossman08}.

We briefly reconcile size conventions. A circuit in our model with $G$ non-input gates can be represented, after sharing input negations and deleting repeated wires, with at most $H:=G+2M+2$ nodes and $H^2$ wires, at depth at most $d+2$. Applying the cited bound at depth $d+2$ gives $H^2\ge N^{b_dr}$ for some $b_d>0$, whether size counts gates or wires. Choose $r_d$ large enough that $b_dr_d/2>2$ and to meet Beame's lower threshold on $r$. Subtracting $2M+2=O(N^2)$ from $H\ge N^{b_dr/2}$ and decreasing the exponent coefficient gives constants $a_d>0$ and $N_d$, depending only on $d$, such that
\begin{equation}\label{eq:beame-ordinary-clique}
\size_d(K_r\text{-}\CLIQUE_N)\ge N^{a_dr}
\qquad(N\ge N_d,\ r_d\le r\le\log N).
\end{equation}

\emph{Transfer to colored clique.} Make $r$ copies of the ordinary graph's vertex set, one for each color, and copy adjacency between different colors. For every pattern edge $f$ with its fixed endpoint order, use the substitution
\begin{equation}\label{eq:ordinary-to-colored-clique}
X_{f,u,w}\longmapsto
\begin{cases}
y_{\{u,w\}},&u\ne w,\\
0,&u=w.
\end{cases}
\end{equation}
The diagonal zeros force the chosen host indices to be distinct, and all remaining required edges test adjacency in the ordinary graph. Thus every colored witness gives an ordinary $r$-clique; conversely, assigning the vertices of an ordinary $r$-clique to the $r$ colors gives a colored witness. Each coordinate is an input bit or a constant, so this is a depth-zero projection $K_r\text{-}\CLIQUE_N\le_{\proj}K_r\text{-}\SUB_N$, preserving $N$ and $r$. By \hyperref[par:projection-monotonicity]{projection monotonicity} and \eqref{eq:beame-ordinary-clique},
\begin{equation*}
\size_d(K_r\text{-}\SUB_N)\ge\size_d(K_r\text{-}\CLIQUE_N)\ge N^{a_dr}.
\end{equation*}
This proves \Cref{lem:beame} with the same constants and range.
\end{proof}

\section{Pattern Constructions and Constants}\label{app:patterns}

We construct the two source pattern families used in the lower bounds, adding disjoint matching edges to reach exactly $k-1$ edges. The expander construction proves \Cref{lem:expander-pattern-family,eq:expander-pattern-kappa}; the clique family supplies the source for \Cref{thm:C}. We use the matching $J_t$ and endpoint orders of \Cref{sec:shared-patterns}.

\subsection{The bounded-degree expander family}\label{app:expander-patterns}

We build a bipartite $4$-regular core with a constant spectral gap, which makes its hardness parameter $\kappa$ linear in its number of vertices. Adding matching edges preserves this lower bound because the core remains a minor of the padded pattern.

\paragraph{Cited bounds.} For a nonempty simple $\Delta$-regular graph $F$ with $\Delta\ge2$, LRR's expansion bound and the discrete Cheeger inequality give the spectral estimate~\cite[Thm.~3.8(ii) and Fact~2.4]{Rosenthal19}
\begin{equation}\label{eq:kappa-spectral}
\kappa(F)\ge\frac{v(F)}6\left(1-\frac{\lambda_2(F)}{\Delta}\right),
\end{equation}
where $\lambda_2(F)$ is the second-largest adjacency eigenvalue. We also use LRR's minor monotonicity and treewidth bound: $\kappa(F)\le\kappa(P)$ when $F$ is a minor of $P$, and $\kappa(P)\le\tw(P)+1$~\cite[Thm.~5.1 and Prop.~4.3]{LRR17}.

\paragraph{The base graphs.}\label{par:expander-threshold} Fix $\Lambda:=2\sqrt3+1/10<4$. Friedman's theorem for uniformly random simple regular graphs~\cite[Thm.~1]{Bordenave20} implies that, for every sufficiently large integer $r$, there is a simple $4$-regular graph on $r$ vertices whose adjacency eigenvalues other than one copy of $4$ have absolute value at most $\Lambda$. Let $n_*\ge5$ be a threshold beyond which this holds. We need only existence: our nonuniform projections fix the pattern before reading the source input.

\begin{proof}[Proof of \Cref{lem:expander-pattern-family} and \Cref{eq:expander-pattern-kappa}]
Choose $k_0\ge\max\{4n_*+1,600\}$. For $k\ge k_0$, put
\begin{equation}\label{eq:degree-four-rounding}
r_k:=\left\lfloor\frac{k-1}{4}\right\rfloor,
\qquad
t_k:=k-1-4r_k\in\{0,1,2,3\}.
\end{equation}
Choose a base graph $Q_k$ as above on $r_k\ge n_*$ vertices. Its bipartite double cover has sides $I_k:=V(Q_k)\times\{0\}$ and $I'_k:=V(Q_k)\times\{1\}$, with edges $\{(a,0),(b,1)\}$ and $\{(b,0),(a,1)\}$ for every edge $\{a,b\}$ of $Q_k$. Call this core $F^{\mathrm{exp}}_k$. It is simple and $4$-regular, with adjacency matrix
\begin{equation}\label{eq:double-cover-adjacency}
\begin{pmatrix}0&A_k\\A_k&0\end{pmatrix},
\end{equation}
where $A_k$ is the adjacency matrix of $Q_k$. Each eigenvector $A_kz=\lambda z$ gives eigenvectors $(z,z)$ and $(z,-z)$ of the cover with eigenvalues $\lambda$ and $-\lambda$. Thus the cover has one eigenvalue $4$, one eigenvalue $-4$, and all remaining eigenvalues in $[-\Lambda,\Lambda]$. The simple eigenvalue $4$ implies connectedness, and $\lambda_2(F^{\mathrm{exp}}_k)\le\Lambda$ gives constant edge expansion by Cheeger's inequality~\cite[Fact~2.4]{Rosenthal19}.

Add the $t_k$ disjoint matching edges, defining
\begin{equation}\label{eq:degree-four-pattern}
P^{\mathrm{exp}}_k:=F^{\mathrm{exp}}_k\mathbin{\dot\cup}J_{t_k}.
\end{equation}
The padded pattern has no isolated vertices, maximum degree $4$, and exactly $k-1$ edges, while
\begin{equation}\label{eq:degree-four-parameters}
v(F^{\mathrm{exp}}_k)=2r_k,
\quad e(F^{\mathrm{exp}}_k)=4r_k>k-5,
\quad |I_k|=|I'_k|=r_k>\frac{k-1}{4}-1.
\end{equation}
By \eqref{eq:kappa-spectral} and minor monotonicity,
\begin{equation}\label{eq:degree-four-kappa-computation}
\kappa(P^{\mathrm{exp}}_k)\ge\kappa(F^{\mathrm{exp}}_k)
\ge\frac{r_k(4-\Lambda)}{12}>\frac{r_k}{30},
\end{equation}
where $\sqrt3<7/4$ gives $4-\Lambda>2/5$. Since $r_k>(k-5)/4$ and $k\ge600$,
\begin{equation}\label{eq:degree-four-constant}
\kappa(P^{\mathrm{exp}}_k)>\frac{k-5}{120}\ge\frac{k-1}{128}.
\end{equation}
This proves \Cref{eq:expander-pattern-kappa} with $c_\kappa=1/128$; $k_0\ge600$ also gives $c_\kappa(k_0-1)>4$.

Finally, fixing all $N^2$ host-edge variables in each matching-edge group to $1$ lets a core witness extend by choosing arbitrary host indices for the new colors. Conversely, every padded witness restricts to a core witness. Keeping the core coordinates unchanged gives the projection
\begin{equation}\label{eq:expander-core-padding-projection}
F^{\mathrm{exp}}_k\text{-}\SUB_N\le_{\proj}P^{\mathrm{exp}}_k\text{-}\SUB_N
\qquad(N\ge1).
\end{equation}
\end{proof}

\paragraph{Treewidth.}\label{par:kappa-versus-treewidth} For our cores $F=F^{\mathrm{exp}}_k$, \eqref{eq:degree-four-kappa-computation} and $\kappa(F)\le\tw(F)+1\le v(F)$ give $\kappa(F)=\Theta(v(F))$.

\paragraph{A limit of regular-core patterns.}\label{par:regular-core-floor} For $P=F\mathbin{\dot\cup}J_t$ with nonempty $\Delta$-regular $F$, $\Delta\ge2$, and $e(P)=k-1$, treewidth is the maximum over connected components. Each matching component has treewidth $1$, and $F$ has at least one edge, so $\tw(P)=\tw(F)$. The upper bound on $\kappa(P)$ is therefore controlled by the core:
\begin{equation}\label{eq:regular-core-kappa-bound}
\kappa(P)\le\tw(P)+1=\tw(F)+1\le v(F)=2e(F)/\Delta\le2(k-1)/\Delta.
\end{equation}
Even a source bound $N^{\kappa(P)}$ would therefore give a leading target exponent of at most $(k-1)/\Delta$ after the quadratic host substitution (\Cref{lem:host-substitutions}), for fixed $k$ as $n\to\infty$. Witness enumeration in \Cref{tab:sigma3-upper-bounds} has leading exponent $k$ in bases $n$ for OV and $n/k$ for XOR at the standard representation scales, so the ratio is $\Delta k/(k-1)>\Delta$, and in particular greater than $4$ for our core. This limitation concerns the LRR bound obtained from regular cores padded by matchings.

\subsection{The clique family}\label{app:clique-patterns}

We add matching edges to the colored-clique pattern to meet the projection requirement $e(P)=k-1$ in \Cref{thm:A}.

\begin{definition}[Padded clique pattern]\label{def:clique-pattern}
For integers $k\ge2$ and $r\ge2$ with $\binom r2\le k-1$, define
\begin{equation}\label{eq:clique-pattern}
P^{\mathrm{clq}}_{k,r}:=K_r\mathbin{\dot\cup}J_{k-1-\binom r2}.
\end{equation}
Then $e(P^{\mathrm{clq}}_{k,r})=k-1$ and $\Delta(P^{\mathrm{clq}}_{k,r})=r-1$.
\end{definition}

\begin{lemma}[Padding a colored clique]\label{lem:clique-source-chain}
For every $k,r$ satisfying \Cref{def:clique-pattern} and every $N\ge1$,
\begin{equation}\label{eq:clique-source-chain}
K_r\text{-}\SUB_N\le_{\proj}P^{\mathrm{clq}}_{k,r}\text{-}\SUB_N.
\end{equation}
\end{lemma}

\begin{proof}
Leave the clique variables unchanged and set all $N^2$ host-edge variables in each added matching-edge group to $1$. Every colored clique then extends by choosing arbitrary host indices for the new colors, while every padded copy restricts to a colored clique, proving the projection.
\end{proof}

\section{Lower Bounds for Other Representation Lengths}\label{app:parameter-completions}

The main results use representations long enough to accommodate a source host of size proportional to $\sqrt n$ for OV or $\sqrt{n/k}$ for XOR and SUM. Shorter representations can still accommodate smaller hosts, giving lower bounds that depend on the available dimension $D$ or bit length $m$.

\subsection{Results at general representation lengths}\label{app:xor-general-rows}

\phantomsection\label{par:sum-high-width}
Let $n\ge k\ge2$ and $\ell_n:=\lceil\log n\rceil$. For each target $Q$, write $h_Q$ for the logarithmic host size before rounding. When $h_Q\ge1$, we use $N_Q:=2^{\lfloor h_Q\rfloor}$ host vertices per color class. For OV, XOR, and SUM, respectively, these parameters satisfy
\begin{equation}\label{eq:D-host-scales}
\begin{aligned}
h_{\mathrm{OV}}&=\Theta\!\left(\min\{\log n,D/k\}\right) &&\text{if }D\ge2(k+2\ell_n),\\
h_\oplus&=\Theta\!\left(\min\{\log(n/k),m/k\}\right) &&\text{if }m\ge16k,\\
h_\Sigma&=\Theta\!\left(\min\{\log(n/k),m/k\}\right) &&\text{if }m\ge18k\log(2k).
\end{aligned}
\end{equation}
For fixed depth $d$ and fixed $k\ge k_0$, the expander bound in \Cref{cor:fixed-source-rate} applies when $N_Q\ge N_0(\gamma,d,k)$, with $\gamma$ chosen as in \eqref{eq:B-rate-slack}. For growing $k$, the depth-three bound requires $k\ge k_D(\gamma)$ and $N_Q\ge k^{C_D(\gamma)}$ for the chosen $0<\gamma<1/2$ (\Cref{thm:D}). Under \Cref{conj:dbd}, the requirements at depth $d$ are $k\ge k_0$ and $N_Q\ge k^{c'_d}$. In each case the exponent in the lower bound is proportional to $kh_Q$, giving $\min\{k\log n,D\}$ for OV and $\min\{k\log(n/k),m\}$ for XOR and SUM up to constant factors in the respective regimes above.

For a clique on $r$ vertices, \Cref{lem:beame} gives an exponent proportional to $r\log N_Q$ when $N_Q\ge N_d$ and $r_d\le r\le\log N_Q$. The pattern has only $k-1$ edges available for the clique, so $\binom r2\le k-1$ as well. These two upper bounds on $r$ account for the factor $\min\{\sqrt{k-1},h_Q\}$ in the clique exponent.

Formally, we define
\begin{equation}\label{eq:D-host-definitions}
\begin{aligned}
h_{\mathrm{OV}}&:=\min\left\{\tfrac12\log n,\frac{D-k-2\ell_n}{4k}\right\},\\
h_\oplus&:=\min\left\{\tfrac12\log\frac{n-1}{k-1},\frac{m}{8k}-1\right\},\\
h_\Sigma&:=\min\left\{\tfrac12\log\frac{n-1}{k-1},\frac{m}{9k}-\log(2k)\right\}.
\end{aligned}
\end{equation}
The first term in each minimum ensures that the encoded objects fit in the arrays or list. The second ensures that each object can encode the required checks, including endpoint agreement. \Cref{lem:D-feasible-hosts} verifies both constraints.

To obtain the estimates in \Cref{eq:D-host-scales} from the definitions in \Cref{eq:D-host-definitions}, observe that the respective length hypotheses give $k+2\ell_n\le D/2$ for OV, $1\le m/(16k)$ for XOR, and $\log(2k)\le m/(18k)$ for SUM. Each subtraction therefore leaves at least half the original quantity. For the list targets, $n/k\le(n-1)/(k-1)\le(n/k)^2$ also puts $\log((n-1)/(k-1))$ within a factor of two of $\log(n/k)$. Taking the minima gives all three estimates with universal constants.

For OV, even a constant-size host uses $2\ell_n$ identifier coordinates, so at fixed $k$ the construction already needs logarithmic dimension. Reducing the host size saves coordinates only on endpoint comparisons, for which the available budget is $D-k-2\ell_n$ (\Cref{eq:ov-dimension-budget}).

Our choice of a power of two gives $\log N_Q=\lfloor h_Q\rfloor\ge h_Q-1$. Thus a host requirement $N_Q\ge k^c$ holds whenever $h_Q\ge c\log k+1$.

Every XOR statement below requires odd $k$. Unless stated otherwise, the growing-$k$ results concern $k=k(n)\to\infty$ with $k=n^{o(1)}$ and hold for all sufficiently large $n$.

\begin{theorem}[Fixed $k$]\label{thm:parameter-fixed}
Let $k_0$ be as in \Cref{lem:expander-pattern-family} and choose $0<2\beta<\gamma<c_\kappa$ as in \eqref{eq:B-rate-slack}. Fix a depth $d$ and an integer $k\ge k_0$. Each bound below holds for its target $Q$ whenever $N_Q\ge N_0(\gamma,d,k)$ and $h_Q\ge\gamma/(\gamma-2\beta)$.
\begin{align*}
\size_d(\kOV_{n,D,k})&\ge2^{2\beta(k-1)h_{\mathrm{OV}}},\\
\size_d(\kXOR_{n,m,k})&\ge2^{2\beta(k-1)h_\oplus},\\
\size_d(\kSUM_{n,m,k})&\ge2^{2\beta(k-1)h_\Sigma}.
\end{align*}
\end{theorem}

\begin{theorem}[Growing $k$ at every fixed depth]
Fix a depth $d$ and take $r_d,N_d$ from \Cref{lem:beame}. In the growing-$k$ regime, put $\rho_Q:=\min\{\sqrt{k-1},h_Q\}$ for each target $Q$. Each bound below holds for its target $Q$ whenever $N_Q\ge N_d$ and $h_Q\ge\max\{2,r_d\}$.
\begin{align*}
\size_d(\kOV_{n,D,k})&\ge\exp\bigl(\Omega_d(\rho_{\mathrm{OV}}h_{\mathrm{OV}})\bigr),\\
\size_d(\kXOR_{n,m,k})&\ge\exp\bigl(\Omega_d(\rho_\oplus h_\oplus)\bigr),\\
\size_d(\kSUM_{n,m,k})&\ge\exp\bigl(\Omega_d(\rho_\Sigma h_\Sigma)\bigr).
\end{align*}
All implied constants depend only on $d$.
\end{theorem}

The expander source gives the following bounds for growing $k$.

\begin{theorem}[Depth three]\label{thm:parameter-sigma-three}
Fix $0<\gamma<1/2$. For every $k\ge k_D(\gamma)$, each of the following bounds holds throughout its displayed host range, with $C_D(\gamma)$ as in \Cref{thm:D}.
\begin{align*}
\size_3(\kOV_{n,D,k})&\ge2^{\gamma k(h_{\mathrm{OV}}-1)} &&\text{if }h_{\mathrm{OV}}\ge C_D(\gamma)\log k+1,\\
\size_3(\kXOR_{n,m,k})&\ge2^{\gamma k(h_\oplus-1)} &&\text{if }h_\oplus\ge C_D(\gamma)\log k+1,\\
\size_3(\kSUM_{n,m,k})&\ge2^{\gamma k(h_\Sigma-1)} &&\text{if }h_\Sigma\ge C_D(\gamma)\log k+1.
\end{align*}
In particular, $\gamma=1/8$ and $C_D(\gamma)=2$ give $\size_3(Q)\ge2^{k(h_Q-1)/8}$ whenever $h_Q\ge2\log k+1$ and $k$ is sufficiently large.
\end{theorem}

\begin{corollary}[Conditional bounds at every fixed depth]
Assume \Cref{conj:dbd}. For every fixed depth $d$, let $\gamma_d,c'_d>0$ be the conjectural source constants. Then
\begin{align*}
\size_d(\kOV_{n,D,k})&\ge2^{\gamma_d k(h_{\mathrm{OV}}-1)} &&\text{if }h_{\mathrm{OV}}\ge c'_d\log k+1,\\
\size_d(\kXOR_{n,m,k})&\ge2^{\gamma_d k(h_\oplus-1)} &&\text{if }h_\oplus\ge c'_d\log k+1,\\
\size_d(\kSUM_{n,m,k})&\ge2^{\gamma_d k(h_\Sigma-1)} &&\text{if }h_\Sigma\ge c'_d\log k+1.
\end{align*}
\end{corollary}

At the representation length $Ck\log n$ of \Cref{tab:intro-frontier}, a sufficiently large $C$ makes the number of entries in each array or list the limiting constraint. Thus $N_Q$ is again of order $\sqrt n$ for OV and $\sqrt{n/k}$ for XOR and SUM, recovering the bounds in the table for fixed $k$ and for growing $k=n^{o(1)}$.

\subsection{Choosing the source size \texorpdfstring{$N$}{N}}\label{app:ov-interpolation}

The following variant of \Cref{lem:host-substitutions} chooses a source size whose projection fits in the target arrays or list at the prescribed representation length. It applies to both source pattern families used in the main text.

\begin{lemma}[Feasible hosts and transfer]\label{lem:D-feasible-hosts}
Let $n\ge k\ge2$ and let $P$ be a simple loopless pattern with $e(P)=k-1$. For any target $Q$ with $h_Q\ge1$ (and odd $k$ for XOR), set $N_Q:=2^{\lfloor h_Q\rfloor}$. Then $P\text{-}\SUB_{N_Q}\le_{\proj}Q$. Consequently, any source bound $\size_d(P\text{-}\SUB_{N_Q})\ge N_Q^a$ with $a>0$ gives
\begin{equation*}
\log\size_d(Q)\ge a\lfloor h_Q\rfloor\ge a(h_Q-1).
\end{equation*}
The same transfer holds separately for $\Sigma_3$ and $\Pi_3$ circuits.
\end{lemma}

\begin{proof}
The required representation lengths in \Cref{eq:ov-dimension-budget,eq:xor-row-budget,eq:sum-packing-parameters} satisfy the following bounds, using $|\Ical(P)|\le2(k-1)$ and $R(P)\le3k$.
\begin{align*}
\Dzero(P,N,n)&\le k+4k\ell_N+2\ell_n,\\
\mxor(P,N)&\le8k\log(2N),\\
\msum(P,N)&\le9k\log(2kN).
\end{align*}
Here $\ell_N\le\log(2N)$ and $\lceil\log(2kN+1)\rceil\le3\log(2kN)$. At $N=N_Q$, we have $\ell_{N_Q}=\log N_Q=\lfloor h_Q\rfloor\le h_Q$. The first term in each minimum gives $N_{\mathrm{OV}}^2\le n$ for OV and $1+(k-1)N_Q^2\le n$ for either list target. The second terms give
\begin{equation*}
k+4k\ell_{N_{\mathrm{OV}}}+2\ell_n\le D,
\qquad
8k\log(2N_\oplus)\le m,
\qquad
9k\log(2kN_\Sigma)\le m.
\end{equation*}
Thus \Cref{thm:A} gives the projections at the prescribed lengths. The size bounds follow by \hyperref[par:projection-monotonicity]{projection monotonicity} and $\lfloor h_Q\rfloor\ge h_Q-1$.
\end{proof}

\paragraph{Applying the source bounds.} For fixed $k$, $N_Q\ge N_0(\gamma,d,k)$ lets us apply \Cref{cor:fixed-source-rate} with source exponent $\gamma(k-1)$. The hypothesis $h_Q\ge\gamma/(\gamma-2\beta)$ ensures $\gamma(h_Q-1)\ge2\beta h_Q$, yielding \Cref{thm:parameter-fixed} by \Cref{lem:D-feasible-hosts}.

For growing $k$ at fixed depth $d$, take a clique on $r=\lfloor\rho_Q\rfloor$ vertices. Then $r\le\lfloor h_Q\rfloor=\log N_Q$ and $\binom r2\le k-1$. Since $h_Q\ge\max\{2,r_d\}$ and $k\to\infty$, eventually $r\ge\max\{2,r_d\}$. With $N_Q\ge N_d$, \Cref{lem:beame} now gives source exponent $a_dr$, and \Cref{lem:clique-source-chain} preserves the bound when we pad the pattern to exactly $k-1$ edges. Since $r\ge\rho_Q/2$ and $\log N_Q\ge h_Q/2$, the transfer yields $\log\size_d(Q)\ge a_d r\log N_Q\ge(a_d/4)\rho_Q h_Q$.

For the remaining expander bounds, $h_Q\ge C\log k+1$ ensures $N_Q=2^{\lfloor h_Q\rfloor}\ge k^C$. At depth three, take $C=C_D(\gamma)$ and apply \eqref{eq:D-source}. Under \Cref{conj:dbd}, take $C=c'_d$ and apply \eqref{eq:uniform-lrr-conjecture}. In each case \Cref{lem:D-feasible-hosts} gives the claimed bound.

\section{A PARITY Projection to Odd \texorpdfstring{$\kXOR$}{k-XOR}}\label{app:xor-supplementary}\label{app:parity-route}

For odd $k\ge3$, $n\ge4k$, $m\ge k\log(n/k)$, and fixed $d\ge2$, a PARITY projection gives
\begin{equation}\label{eq:parity-route-bound}
\size_d(\kXOR_{n,m,k})\ge\exp\!\left(\Omega_d\!\left((k^2\log(n/k))^{1/(d-1)}\right)\right).
\end{equation}
At the standard representation scale $m=\Theta(k\log(n/k))$, this bound is tight at depth two up to constants in the exponent (\Cref{cor:xor-depth-two-tight}). At any fixed depth, its exponent is optimal up to constant factors among lower bounds transferred from PARITY by a literal projection. The construction uses $\ell=\Theta(k^2\log(n/k))$ source bits from PARITY, the largest possible order even with arbitrarily many rows (\Cref{prop:parity-projection-saturation}), and depth-$d$ PARITY has size $\exp(\Theta_d(\ell^{1/(d-1)}))$~\cite{Hastad89}.

For fixed odd $k$ and fixed $d\ge3$, PARITY transfers at most $\exp(O_{d,k}(\sqrt{\log n}))=n^{o(1)}$, while \Cref{thm:B} gives $(n/k)^{\Omega(k)}$ in its range. For a growing example, take odd $k=\Theta(\log^2 n)$ and $m=\Theta(k\log(n/k))$ with a sufficiently large constant factor. At depth four, the strongest lower bound transferable from PARITY is at most $\exp(O(\log^{5/3}n))$, while \Cref{thm:C} gives $\exp(\Omega(\log^2 n))$. The gap persists at every larger fixed depth. At depth three, \Cref{thm:D} gives an exponent of order $\log^3 n$ for both choices of top gate (\Cref{tab:parity-comparison}).

\begin{table}[htbp]
\centering\small
\renewcommand{\arraystretch}{1.15}
\begin{tabular}{@{}lccc@{}}
\toprule
Class & PARITY & Subgraph lower bound & Larger bound \\
\midrule
Depth $2$ & $\Theta(\log^5 n)$ & $\Omega(\log^3 n)$ (\Cref{thm:D}) & PARITY \\
Depth $3$ & $\Theta(\log^{5/2}n)$ & $\Omega(\log^3 n)$ (\Cref{thm:D}) & Subgraph \\
Fixed $d\ge4$ & $\Theta_d(\log^{5/(d-1)}n)$ & $\Omega_d(\log^2 n)$ (\Cref{thm:C}) & Subgraph \\
\bottomrule
\end{tabular}
\caption{Exponents in $\exp(\cdot)$ for the unconditional lower bounds at odd $k=\Theta(\log^2 n)$ and $m=\Theta(k\log(n/k))$ with a sufficiently large constant factor. The depth-two subgraph bound follows from $\size_2\ge\size_3$.
}\label{tab:parity-comparison}
\end{table}

More generally, throughout the range of \Cref{thm:D}, the depth-three subgraph lower bound has exponent $\Omega(k\log(n/k))$, while PARITY gives at most $O(k\sqrt{\log(n/k)})$. Under \Cref{conj:dbd}, \Cref{thm:E} also gives exponent $\Omega_d(\log^3 n)$ for every fixed $d\ge4$ in the example of \Cref{tab:parity-comparison}.

\begin{theorem}[PARITY projection]\label{thm:parity-route}
Let $k\ge3$ be odd and $n\ge4k$. Put
\begin{equation}\label{eq:parity-route-parameters}
b:=\left\lfloor\log\frac{n}{2k}\right\rfloor,\qquad B:=(k-1)b+1,\qquad \ell:=(k-1)B.
\end{equation}
For every $m\ge m_0:=B+k$, there is a projection $\mathsf{PARITY}_\ell\le_{\proj}\kXOR_{n,m,k}$, and \eqref{eq:parity-route-bound} holds at every fixed depth $d\ge2$. Here $m_0\le k\log(n/k)$ and $\ell=\Theta(k^2\log(n/k))$.
\end{theorem}

\begin{proof}
Write the source input as $k-1$ vectors $x_1,\ldots,x_{k-1}\in\bits^B$. The $h$-th bit of $p:=x_1\xor\cdots\xor x_{k-1}$ is the parity of the $h$-th bits of these vectors. Thus the source input has odd parity exactly when $p$ does. This identity holds for both odd and even $k$.

The target input is an $m_0\times n$ matrix whose columns are the $n$ vectors of the $\kXOR$ instance. Partition the columns into $k$ groups: for each $t<k$, group $t$ has $2^b$ columns labeled by $v\in\bits^b$; group $k$ contains the remaining $n-(k-1)2^b>0$ columns, since $2^b\le n/(2k)$.

Use row $t\in[k]$ to test the parity of the number of columns selected from group $t$. Its entries are $0$ on group $t$ and $1$ elsewhere. If the selection takes $c$ of its $k$ columns from group $t$, it has $k-c$ ones in this row. For odd $k$, their XOR is zero exactly when $c$ is odd. Hence a selection with zero XOR in all $k$ of these rows must take at least one column from every group, and with $k$ columns total it takes exactly one from each. Conversely, one column from each group gives zero XOR in all $k$ rows. This is the only step of the construction that uses odd $k$.

Number the remaining $B$ data rows from $1$ to $B$, and partition the first $B-1$ of them into $k-1$ consecutive blocks of $b$ rows. For group $t<k$, let $w_t(v)\in\bits^B$ place the label $v$ in block $t$, its parity $\bigoplus_{j=1}^b v_j$ in row $B$, and zeros elsewhere. Give the column labeled $v$ in group $t$ the data vector $x_t\xor w_t(v)$. Give every column in group $k$ the unit vector $e_B$ with a one in row $B$. Since $w_t(v)$ is fixed by the column label, each matrix entry is a constant or a literal.

Now select one column from each group and let $v_t$ be the label from group $t<k$. The XOR is zero in all the data rows exactly when
\begin{equation}\label{eq:parity-row-condition}
p=e_B\xor w_1(v_1)\xor\cdots\xor w_{k-1}(v_{k-1}).
\end{equation}
In the first $B-1$ coordinates of \eqref{eq:parity-row-condition}, group $t<k$ contributes only its label $v_t$, in block $t$. Equality in these coordinates therefore forces $v_t$ to be the corresponding block of $p$. Each of the first $k-1$ groups contains its required label. With those labels fixed, the vectors $w_t(v_t)$ together contribute the parity of $p_1,\ldots,p_{B-1}$ to the last coordinate, and $e_B$ adds one. The last coordinate of \eqref{eq:parity-row-condition} thus requires
\[
p_B=1\xor p_1\xor\cdots\xor p_{B-1}.
\]
This equation holds exactly when $p_1\xor\cdots\xor p_B=1$. If it holds, the columns with the required labels and any column from group $k$ give a zero-XOR selection. If it fails, changing a label would violate one of the first $B-1$ data rows. Hence a zero-XOR selection exists exactly when $p$ has odd parity, which is equivalent to odd parity of the source input.

We used $B+k=(k-1)(b+1)+2\le k\log(n/k)$ rows, since $b+1\le\log(n/k)$ and $n/k\ge4$. Appending zero rows gives every larger $m$. Finally, $b=\Theta(\log(n/k))$ for $n/k\ge4$, giving the claimed source length. By projection monotonicity, H\aa stad's depth-$d$ lower bound of $\exp(\Omega_d(\ell^{1/(d-1)}))$ for PARITY gives \eqref{eq:parity-route-bound}.
\end{proof}

\begin{corollary}[Tight depth-two bound for odd $\kXOR$]\label{cor:xor-depth-two-tight}
For odd $k\ge3$, $n\ge4k$, and $m\ge k\log(n/k)$, one has $\size_2(\kXOR_{n,m,k})\ge 2^{\Omega(k^2\log(n/k))}$. If also $m=O(k\log(n/k))$, then
\begin{equation}\label{eq:parity-depth-two-lower}
\size_2(\kXOR_{n,m,k})=2^{\Theta(k^2\log(n/k))}.
\end{equation}
\end{corollary}

\begin{proof}
The projection in \Cref{thm:parity-route} transfers the classical depth-two PARITY bound of $2^{\ell-1}$ terms in a DNF or clauses in a CNF~\cite[Lemma~2]{EmdinKMS22}.

For the upper bound, fix a support of size $k$. The first $k-1$ selected vectors uniquely determine the last vector that makes their XOR zero. Each of the $2^{(k-1)m}$ resulting assignments is checked by one AND gate. An OR over these terms for all supports gives a DNF with at most $1+\binom nk2^{(k-1)m}$ gates. The logarithm of this count is $O(k^2\log(n/k))$ when $m=O(k\log(n/k))$.
\end{proof}

\begin{corollary}[Maximum PARITY source length]\label{prop:parity-projection-saturation}
For any $n\ge k\ge1$ and any $m$, a projection $\mathsf{PARITY}_\ell\le_{\proj}\kXOR_{n,m,k}$ satisfies
\begin{equation}\label{eq:parity-saturation-bound}
\ell\le k\left(1+\log\binom nk\right)=O\bigl(k^2(1+\log(n/k))\bigr).
\end{equation}
\end{corollary}

\begin{proof}
For a fixed choice of $k$ target columns, the XOR in each row depends on at most $k$ source bits. The test that this XOR is zero has a width-$k$ CNF, with one clause excluding each falsifying assignment to those bits. Conjoining the tests gives a width-$k$ CNF for the event that the choice is a witness. Thus $\mathsf{PARITY}_\ell$ is the disjunction of $\binom nk$ such CNFs. Emdin et al.~\cite[\S3.2]{EmdinKMS22} show via the Satisfiability Coding Lemma~\cite{PaturiPudlakZane99} that this requires at least $2^{\ell/k-1}$ CNFs. Hence $2^{\ell/k-1}\le\binom nk$, and $\binom nk\le(en/k)^k$ gives the asymptotic estimate.
\end{proof}


\section{Upper Bounds and Comparison with Lower Bounds}\label{app:upper-bounds}

We give elementary $\Sigma_3$ circuits and compare their sizes with the lower bounds for unrestricted depth three at the standard representation scales of \Cref{sec:fixed-target-bounds}.

The constructions below apply to all admissible parameters with $k\ge2$. Each circuit uses a top OR over candidate witnesses and one CNF per candidate.

\paragraph{OV.} For each of the $n^k$ tuples of indices, we use one CNF with a clause for each coordinate, checking that at least one selected vector has a zero there. Each CNF has $D$ OR gates and one AND gate, for a total of $1+n^k(D+1)$ gates including the top OR.

\begin{table}[htbp]
\centering
\small
\renewcommand{\arraystretch}{1.08}
\setlength{\tabcolsep}{3pt}
\begin{tabular}{@{}>{\raggedright\arraybackslash}p{0.14\textwidth}>{\raggedright\arraybackslash}p{0.31\textwidth}>{\raggedright\arraybackslash}p{0.23\textwidth}>{\raggedright\arraybackslash}p{0.26\textwidth}@{}}
\toprule
Target & Representation length & Fixed-$k$ upper bound & Growing-$k$ upper bound \\
\midrule
$\kOV$ & $D=\Theta(k\log n)$ & $n^{k+o(1)}$ & $n^{k+o(k)}$ \\
\addlinespace[2pt]
$\kXOR$ & $m=\Theta(k\log(n/k))$ & $(n/k)^{k+o(1)}$ & $(n/k)^{k+o(k)}$ \\
\addlinespace[2pt]
$\kSUM$ & $m=\Theta(k\log(n/k))$ & $(n/k)^{k+o(1)}$ & $(n/k)^{O(k)}$ \\
\bottomrule
\end{tabular}
\caption{Circuit-size upper bounds at the standard representation scales. All circuits have depth three with a top OR gate ($\Sigma_3$). The growing-$k$ column assumes $n\ge k^C$ for a fixed $C>1$.}
\label{tab:sigma3-upper-bounds}
\end{table}

\FloatBarrier

\paragraph{XOR.} For each of the $\binom{n}{k}$ supports, every row must have even parity. Use $2^{k-1}$ clauses per row, one to exclude each odd assignment to the selected bits. Conjoin the clauses for all $m$ rows to get one CNF with $m2^{k-1}$ clauses for the support. The circuit has $1+\binom{n}{k}(1+m2^{k-1})$ gates.

\paragraph{SUM.} Put $b=\lceil\log k\rceil$. For each of the $\binom{n}{k}$ supports, split the $m$-bit representations of the selected residues into $t=\lceil m/b\rceil$ blocks, starting at the least significant end. Write $b_j$ for the length of block $j$, so $b_j=b$ except possibly in the final block. Let $s_j$ be the sum of the $k$ integers represented by the selected blocks at position $j$, and let $c_j$ be the carry into that block. Dividing $s_j+c_j$ by $2^{b_j}$ gives the output value $r_j$ as the remainder and the carry $c_{j+1}$ as the quotient, so
\[
s_j+c_j=r_j+2^{b_j}c_{j+1},\qquad 0\le r_j<2^{b_j}.
\]
The addition starts with $c_1=0$. Each of the $k$ block values is at most $2^{b_j}-1$, so if $c_j\le k-1$ then $s_j+c_j\le k(2^{b_j}-1)+(k-1)<k2^{b_j}$. The quotient $c_{j+1}$ is therefore at most $k-1$, and induction gives this bound for every carry.

For a nonfinal block with zero output, the equation reduces to $s_j+c_j=2^b c_{j+1}$. Since $0\le c_j\le k-1<2^b$, we have $s_j/2^b\le c_{j+1}<s_j/2^b+1$. The only integer in this interval is $\lceil s_j/2^b\rceil$. Thus a block that outputs zero determines the carry into the next block from its own input bits.

The final carry contributes a multiple of $2^m$, so the modular sum is zero exactly when all output blocks are zero. The check for block $j$ tests whether $s_j+\widehat c_j$ is divisible by $2^{b_j}$, where $\widehat c_1=0$ and $\widehat c_j=\lceil s_{j-1}/2^b\rceil$ for $j>1$. At the first block, $\widehat c_1=c_1=0$, and the carry formula gives $\widehat c_j=c_j$ after any zero-output block. If the modular sum is zero, each check therefore passes with remainder $r_j=0$. If the modular sum is nonzero, let $j$ be the first block with $r_j\ne0$. All earlier output blocks are zero, so $\widehat c_j=c_j$. Hence $(s_j+\widehat c_j)\bmod 2^{b_j}=r_j\ne0$, and the divisibility check fails.

For each of the $\binom{n}{k}$ supports, use one CNF to enforce all $t$ block checks. A check reads at most two adjacent blocks per selected residue, or at most $k\min\{m,2b\}$ bits. Exclude each failing assignment to these bits with one clause, giving at most $t2^{k\min\{m,2b\}}$ clauses per CNF. Including one AND gate per CNF and the top OR, the circuit has at most $1+\binom{n}{k}\bigl(1+t2^{k\min\{m,2b\}}\bigr)$ gates.

\paragraph{Comparison.} At the standard representation scales and fixed $k$, each CNF uses $O(\log n)$ gates. For growing $k$ in the range of \Cref{tab:sigma3-upper-bounds}, the CNFs for OV and XOR have $n^{o(k)}$ and $(n/k)^{o(k)}$ gates, respectively. Since $t\le m$, each SUM CNF has at most $1+m2^{2kb}$ gates, whose logarithm is $O(k\log k+\log m)$. In the range $n\ge k^C$ we have $\log k\le\log(n/k)/(C-1)$, and the standard width gives $\log m\le m=O(k\log(n/k))$. Thus a SUM CNF has at most $2^{O(k\log(n/k))}=(n/k)^{O(k)}$ gates. Together with $\binom{n}{k}\le(en/k)^k$, these estimates give the bounds in the table. Under their respective hypotheses, \Cref{thm:B,thm:D} match these upper bounds up to constant factors in the exponent.

\end{document}